\documentclass{article}

\usepackage[left=1.25in,top=1.25in,right=1.25in,bottom=1.25in,head=1.25in]{geometry}
\usepackage{amsfonts,amsmath,amssymb,amsthm}
\usepackage{graphicx,psfrag}
\usepackage[round,colon]{natbib}

\usepackage{calc}
\usepackage{fancyheadings}
\usepackage{ifthen}
\usepackage{tabularx}

\def\mR{\mathbb{R}}
\def\mE{\mathbb{E}}

\def\hs{\hat{\sigma}}

\def\hf{\hat{f}}

\def\Err{{\rm Err}}
\def\ErrF{{\rm ErrF}}
\def\ErrS{{\rm ErrS}}
\def\ErrR{{\rm ErrR}}

\def\OptF{{\rm OptF}}
\def\OptS{{\rm OptS}}
\def\OptR{{\rm OptR}}

\def\eb{\text{$B^+$}}

\def\ev{\text{$V^+$}}
\def\rcp{{\rm RCp}}
\def\rcpp{\text{{\rm RCp}$^+$}}
\def\hrcp{\text{$\widehat{{\rm RCp}}$}}

\def\Cov{\mathrm{Cov}}
\def\Var{\mathrm{Var}}
\def\tr{\mathrm{tr}}

\newcommand{\argmin}{\mathop{\mathrm{argmin}}}

\newcommand{\comment}[1]{}

\newtheorem{thm}{Theorem}

\newtheorem{cor}{Corollary}
\newtheorem{prop}{Proposition}

\title{From Fixed-X to Random-X Regression: Bias-Variance
  Decompositions, Covariance Penalties, and Prediction Error
  Estimation}
\author{Saharon Rosset \and Ryan J. Tibshirani
\thanks{The authors thank Edgar Dobriban for help in formulating and proving Theorem 3, and Felix Abramovich, Trevor Hastie, Amit Moscovich, Moni Shahar, Rob Tibshirani and Stefan Wager for useful comments.}}
\date{}

\begin{document}
\maketitle

\begin{abstract}
In statistical prediction, classical approaches for model selection and model evaluation based on covariance penalties are still widely used. Most of the literature on this topic is based on what we call the ``Fixed-X'' assumption, where covariate values are assumed to be nonrandom. By contrast, it is often more reasonable to take a ``Random-X'' view, where the covariate values are independently drawn for both training and prediction. To study the applicability of covariance penalties in this setting, we propose a decomposition of Random-X prediction error in which the randomness in the covariates contributes to both the bias and variance components.  This decomposition is general, but we concentrate on the fundamental case of least squares regression. We prove that in this setting the move from Fixed-X to Random-X prediction results in an increase in both bias and variance. When the covariates are normally distributed and the linear model is unbiased, all terms in this decomposition are explicitly computable, which yields an extension of Mallows' Cp that we call \rcp. \rcp{} also holds asymptotically for certain classes of nonnormal covariates.
 When the noise variance is unknown, plugging in the usual unbiased estimate leads to an approach that we call \smash{\hrcp}, which is closely related to Sp (Tukey 1967), and GCV (Craven and Wahba 1978). For excess bias, we propose an estimate based on the ``shortcut-formula'' for ordinary cross-validation (OCV), resulting in an approach we call \rcpp. Theoretical arguments and numerical simulations suggest that \rcpp{} is typically superior to OCV, though the difference is small.  We further examine the Random-X error of other popular estimators. The surprising result we get for ridge regression is that, in the heavily-regularized regime, Random-X variance is smaller than Fixed-X variance, which can lead to smaller overall Random-X error.
\end{abstract}

\section{Introduction}

A statistical regression model seeks to describe the relationship
between a response $y \in \mR$ and a covariate vector $x \in \mR^p$,
based on training data comprised of paired observations
$(x_1,y_1),\ldots,(x_n,y_n)$. Many modern regression models are
ultimately aimed at prediction: given a new covariate value $x_0$,
we apply the model to predict the corresponding response value
$y_0$. Inference on the prediction error of
regression models is a central part of model evaluation and model
selection in statistical learning (e.g., \citealt{hastie2009elements}).
A common assumption that is used in the estimation of
prediction error is what we call a ``Fixed-X'' assumption,
where the training covariate values $x_1,\ldots,x_n$ are treated as
fixed, i.e., nonrandom, as are the covariate values at which
predictions are to be made, $x_{01},\ldots,x_{0n}$, which are also
assumed to equal the training values. In the Fixed-X setting,
the celebrated notions of optimism and degrees of freedom lead to
covariance penalty approaches to estimate the prediction performance
of a model
\citep{efron1986biased,efron2004estimation,hastie2009elements},
extending and generalizing classical approaches like Mallows' Cp
\citep{mallows1973comments} and AIC \citep{akaike1973information}.

The Fixed-X setting is one of the most common views on
regression (arguably the predominant view), and it can be found at
all points on the spectrum from cutting-edge research to
introductory teaching in statistics. This setting
combines the following two assumptions about
the problem.
\begin{enumerate}
\item[(i)] The covariate values $x_1,\ldots,x_n$ used in training are not
  random (e.g., designed), and the only randomness in training is due
  to the responses $y_1,\ldots,y_n$.
\item[(ii)] The covariates $x_{01},\ldots,x_{0n}$ used for
  prediction exactly match $x_1,\ldots,x_n$, respectively,
  and the corresponding responses $y_{01},\ldots,y_{0n}$ are independent
  copies of $y_1,\ldots,y_n$, respectively.
\end{enumerate}
Relaxing assumption (i), i.e., acknowledging randomness in the
training covariates $x_1,\ldots,x_n$, and taking this randomness
into account when performing inference on estimated parameters and
fitted models, has received a good deal of attention in the
literature. But, as we see it, assumption (ii) is the critical one
 that needs to be relaxed in most realistic prediction
setups.   To emphasize this, we define two settings beyond the Fixed-X
one, that we call the ``Same-X'' and ``Random-X'' settings.  The
Same-X setting drops assumption (i), but does not account for new
covariate
values at prediction time.  The Random-X setting drops both
assumptions, and deals with predictions at new covariates
values.  These will be defined more precisely in the next subsection.

\subsection{Notation and assumptions}

We assume that the training data
$(x_1,y_1),\ldots,(x_n,y_n)$ are i.i.d.\ according to some joint
distribution $P$. This is an innocuous assumption, and it
means that we can posit a relationship for the training data,
\begin{equation}
\label{eq:iid_pairs}
y_i = f(x_i) + \epsilon_i, \quad i=1,\ldots,n
\end{equation}
where $f(x)=\mE(y|x)$, and the expectation here is taken with respect
to a draw $(x,y) \sim P$.  We also assume that for $(x,y) \sim P$,
\begin{equation}
\label{eq:eps_indep_x}
\text{$\epsilon=y-f(x)$ is independent of $x$},
\end{equation}
which is less innocuous, and precludes, e.g.,
heteroskedasticity in the data.  We let $\sigma^2 = \Var(y|x)$ denote
the constant conditional variance.  It is worth pointing out that some
results in this paper can be adjusted or modified to hold when
\eqref{eq:eps_indep_x} is not assumed; but since other results hinge
critically on \eqref{eq:eps_indep_x}, we find it is more convenient
to assume \eqref{eq:eps_indep_x} up front.

For brevity, we write $Y=(y_1,\ldots,y_n) \in \mR^n$ for the vector of
training responses, and $X \in \mR^{n\times p}$ for the matrix of
training covariates with $i$th row $x_i$, $i=1,\ldots,n$.
We also write $Q$ for the marginal distribution of $x$ when $(x,y)
\sim P$, and $Q^n=Q \times \cdots \times Q$ ($n$ times) for the
distribution of $X$ when its $n$ rows are drawn i.i.d.\ from $Q$.
We denote by \smash{$\tilde{y}_i$} an independent copy of $y_i$, i.e.,
an independent draw from the conditional law of $y_i|x_i$, for
$i=1,\ldots,n$, and we abbreviate
\smash{$\tilde{Y}=(\tilde{y}_1,\ldots,\tilde{y}_n) \in \mR^n$}.
These are the responses considered in the Same-X setting,
defined below.
We denote by $(x_0,y_0)$ an independent draw from $P$.  This the
covariate-response pair evaluated in the Random-X setting, also
defined below.

Now consider a model building procedure that uses the training data
$(X,Y)$ to build a prediction function \smash{$\hf_n : \mR^p \to
  \mR$}.  We can associate to this procedure two notions of prediction
error:
$$
\ErrS = \mE_{X,Y,\tilde{Y}} \bigg[\frac{1}{n} \sum_{i=1}^n
\big( \tilde{y}_i - \hf_n(x_i)\big)^2\bigg]
\quad \text{and} \quad
\ErrR = \mE_{X,Y,x_0,y_0} \big( y_0 - \hf_n(x_0) \big)^2,
$$
where the subscripts on the expectations highlight the random
variables over which expectations are taken.  (We omit subscripts when
the scope of the expectation is clearly understood by the context.)
The {\it Same-X} and {\it Random-X} settings
differ only in the quantity we use to measure
prediction error: in Same-X, we use \ErrS, and in Random-X, we use
\ErrR.  We call \ErrS{} the Same-X prediction error and
\ErrR{} the Random-X prediction error, though we note these are also
commonly called in-sample and out-of-sample prediction
error, respectively.  We also note that by exchangeability,
$$
\ErrS = \mE_{X,Y,\tilde{y}_1} \big( \tilde{y}_1 - \hf_n(x_1)\big)^2.
$$
Lastly, the {\it Fixed-X} setting is defined by the same model
assumptions as above, but with $x_1,\ldots,x_n$ viewed as nonrandom,
i.e., we assume the responses are drawn from \eqref{eq:iid_pairs},
with the errors being i.i.d. We can equivalently view this as the
Same-X setting, but where we condition on $x_1,\ldots,x_n$. In
the Fixed-X setting, prediction error is defined by
$$
\ErrF = \mE_{Y,\tilde{Y}} \bigg[ \frac{1}{n} \sum_{i=1}^n
\big( \tilde{y}_i - \hf_n(x_i)\big)^2\bigg].
$$
(Without $x_1,\ldots,x_n$ being random, the terms in the
sum above are no longer exchangeable, and so \ErrF{} does
not simplify as \ErrS{} did.)

\subsection{Related work}

From our perpsective, much of the work encountered in statistical
modeling takes a Fixed-X view, or when treating the covariates as
random, a Same-X view.  Indeed, when concerned with parameter
estimates and parameter inferences in regression models, the
randomness of new prediction points plays no role, and so the
Same-X view seems entirely appropriate.  But, when focused on
prediction, the Random-X view seems more realistic as a study ground
for what happens in most applications.

On the other hand, while the Fixed-X view is common, the Same-X and
Random-X views have not exactly been ignored, either, and several
groups of researchers in statistics, but also in machine learning and
econometrics, fully adopt and argue for such random covariate views.
A scholarly and highly informative treatment of how randomness in the
covariates affects parameter estimates and inferences in regression
models is given in \citet{buja2014models,buja2016models}.  We also
refer the reader to these papers for a nice review of the history of
work in statistics and econometrics on random covariate models.
It is also worth mentioning that in nonparametric regression theory,
it is common to treat the covariates as random, e.g., the book by
\citet{gyorfi2002distribution}, and the random covariate view is the
standard in what machine learning researchers call statistical
learning theory, e.g., the book by \citet{vapnik1998statistical}.
Further, a stream of recent papers in high-dimensional regression
adopt a random covariate perspective, to give just a few examples:
\citet{greenshtein2004persistence,chatterjee2013assumptionless,
dicker2013optimal,hsu2014analysis,dobriban2015high}.

In discussing statistical models with random covariates, one should
differentiate between what may be called the ``i.i.d.\ pairs'' model
and ``signal-plus-noise'' model.  The former assumes i.i.d.\ draws
$(x_i,y_i)$, $i=1,\ldots,n$ from a common distribution $P$, or
equivalently i.i.d.\ draws from the model \eqref{eq:iid_pairs}; the
latter assumes i.i.d.\ draws from \eqref{eq:iid_pairs}, and
additionally assumes \eqref{eq:eps_indep_x}.  The additional
assumption \eqref{eq:eps_indep_x} is not a light one, and it does not
allow for, e.g., heteroskedasticity.
The books by
\citet{vapnik1998statistical,gyorfi2002distribution} assume the
i.i.d.\ pairs model, and do not require \eqref{eq:eps_indep_x}
(though their results often require a bound on the maximum of
$\Var(y|x)$ over all $x$.)

More specifically related to the focus of our paper is the seminal
work of \citet{breiman1992submodel}, who considered Random-X
prediction error mostly from an intuitive and empirical point of view.
A major line of work on practical covariance penalties for Random-X
prediction error in least squares regression begins with
\citet{stein1960multiple} and \citet{tukey1967discussion}, and
continues onwards throughout the late 1970s and early 1980s with
\citet{hocking1976analysis,thompson1978selection1,
thompson1978selection2,breiman1983how}. Some more recent contributions
are found in \citet{leeb2008evaluation,dicker2013optimal}.
A common theme to these works is the assumption that $(x,y)$
is jointly normal.  This is a strong assumption, and is one that we
avoid in our paper (though for some results we assume $x$ is
marginally normal); we will discuss comparisons to these works later.
Through personal communication, we are aware of work in progress
by Larry Brown, Andreas Buja, and coauthors on a variant of Mallows'
Cp for a setting in which covariates are random.  It is out
understanding that they take somewhat of a broader view than we do in
our proposals \smash{$\rcp,\hrcp,\rcpp$}, each designed for a more
specific scenario, but resort to asymptotics in order to do so.

Finally, we must mention that an important alternative to covariance
penalties for Random-X model evaluation and selection are
resampling-based techniques, like cross-validation and bootstrap
methods (e.g.,
\citealt{efron2004estimation,hastie2009elements}). In particular,
ordinary leave-one-out cross-validation or OCV evaluates a model by
actually building $n$ separate prediction models, each one using $n-1$
observations for training, and one held-out observation for model
evaluation. OCV naturally provides an almost-unbiased estimate
of Random-X prediction error of a modeling approach (``almost'',
since training set sizes are $n-1$ instead of $n$), albeit, at a
somewhat high price in terms of variance and inaccuracy (e.g., see
\citealt{burman1989comparative,hastie2009elements}).
Altogether, OCV is an important benchmark for comparing the results of
any proposed Random-X model evaluation approach.

\section{Decomposing and estimating prediction error}
\label{sec:decomp_pred_error}

\subsection{Bias-variance decompositions}

Consider first the Fixed-X setting, where $x_1,\ldots,x_n$ are
nonrandom. Recall the well-known decomposition of Fixed-X
prediction error (e.g., \citealt{hastie2009elements}):
$$
\ErrF = \sigma^2 + \frac{1}{n} \sum_{i=1}^n
\big(\mE\hf_n(x_i) - f(x_i)\big)^2 +
\frac{1}{n} \sum_{i=1}^n \Var\big(\hf_n(x_i)\big)
$$
where the latter two terms on the right-hand side above are called the
(squared) {\it bias} and {\it variance} of the estimator
\smash{$\hf_n$}, respectively.   In the Same-X setting, the same
decomposition holds conditional on $x_1,\ldots,x_n$.  Integrating out
over $x_1,\ldots,x_n$, and using exchangeability, we conclude
$$
\ErrS = \sigma^2 + \underbrace{\mE_X
  \Big(\mE\big(\hf_n(x_1)\,|\,X\big) - f(x_1)\Big)^2}_{B} +
\underbrace{\mE_X \Var\big(\hf_n(x_1)\,|\,X\big)
\vphantom{\Big(\Big)}}_{V}.
$$
The last two terms on the right-hand side above are integrated bias
and variance terms associated with \smash{$\hf_n$}, which we denote by
$B$ and $V$, respectively.  Importantly, whenever the Fixed-X variance
of the estimator \smash{$\hf_n$} in question is unaffected by the
form of $f(x)=\mE(y|x)$ (e.g., as is the case in least squares
regression), then so is the integrated variance $V$.

For Random-X, we can condition on $x_1,\ldots,x_n$ and $x_0$, and then
use similar arguments to yield the decomposition
$$
\ErrR = \sigma^2 + \mE_{X,x_0}\Big(
\mE\big(\hf_n(x_0)\,|\,X,x_0\big) - f(x_0)\Big)^2 +
\mE_{X,x_0} \Var\big(\hf_n(x_0)\,|\,X,x_0\big).
$$
For reasons that will become clear in what follows,
it suits our purpose to rearrange this as
\begin{align}
\label{eq:er}
\ErrR &= \sigma^2 + B + V \\
\label{eq:eb}
&\quad +
\underbrace{\mE_{X,x_0}\Big(\mE\big(\hf_n(x_0)\,|\,X,x_0\big) -
  f(x_0)\Big)^2 - \mE_X \Big(\mE\big(\hf_n(x_1)\,|\,X\big) -
  f(x_1)\Big)^2}_{\eb} \\
\label{eq:ev}
&\quad +
\underbrace{\mE_{X,x_0} \Var\big(\hf_n(x_0)\,|\,X,x_0\big) -
\mE_X \Var\big(\hf_n(x_1)\,|\,X\big)}_{\ev}.
\end{align}
We call the quantities in \eqref{eq:eb}, \eqref{eq:ev} the {\it
  excess bias} and {\it excess variance} of \smash{$\hf_n$}
(``excess'' here referring to the extra amount of bias and variance
that can be attributed to the randomness of $x_0$), denoted by \eb{}
and \ev{}, respectively. We note that, by construction,
$$
\ErrR - \ErrS = \eb + \ev,
$$
thus, e.g., $\eb+\ev \geq 0$ implies the Random-X
(out-of-sample) prediction error of \smash{$\hf_n$} is no smaller than
its Same-X (in-sample) prediction error. Moreover, as $\ErrS$ is
easily estimated following standard practice for estimating $\ErrF$,
discussed next, we see that estimates or bounds $\eb,\ev$
lead to estimates or bounds on $\ErrR$.

\subsection{Optimism for Fixed-X and Same-X}

Starting with the Fixed-X setting again, we recall the definition of
optimism, e.g., as in \citet{efron1986biased,efron2004estimation,
  hastie2009elements},
$$
\OptF = \mE_{Y,\tilde{Y}} \bigg[\frac{1}{n} \sum_{i=1}^n
\big(\tilde{y}_i-\hf_n(x_i)\big)^2 -
\frac{1}{n} \sum_{i=1}^n \big(y_i-\hf_n(x_i)\big)^2\bigg],
$$
which is the difference in prediction error and training error.
Optimism can also be expressed as the following elegant sum of
self-influence terms,
$$
\OptF = \frac{2}{n} \sum_{i=1}^n \Cov \big(y_i, \hf_n(x_i)\big),
$$
and furthermore, under a normal regression model (i.e., the data model
\eqref{eq:iid_pairs} with $\epsilon \sim N(0,\sigma^2)$) and some
regularity conditions on \smash{$\hf_n$} (i.e., continuity and almost
differentiability as a function of $y$),
$$
\OptF = \frac{2\sigma^2}{n} \sum_{i=1}^n \mE \bigg[\frac{\partial
  \hf_n(x_i)}{\partial y_i}\bigg],
$$
which is often called Stein's formula \citep{stein1981estimation}.

Optimism is an interesting and important concept because an unbiased
estimate \smash{$\widehat{\OptF}$} of \OptF{} (say, from Stein's
formula or direct calculation) leads to an unbiased estimate of
prediction error:
$$
\frac{1}{n} \sum_{i=1}^n \big(y_i-\hf_n(x_i)\big)^2 + \widehat{\OptF}.
$$
When \smash{$\hf_n$} is given by the least squares regression
of $Y$ on $X$ (and $X$ has full column rank), so that
\smash{$\hf_n(x_i)=x_i^T (X^T X)^{-1} X^T Y$}, $i=1,\ldots,n$, it is
not hard to check that \smash{$\OptF=2\sigma^2p/n$}. This is exact
and hence ``even better'' than an unbiased estimate; plugging in this
result above for \smash{$\widehat{\OptF}$} gives us Mallows'
Cp \citep{mallows1973comments}.

In the Same-X setting, optimism can be defined similarly, except
additionally integrated over the distribution of $x_1,\ldots,x_n$,
$$
\OptS = \mE_{X,Y,\tilde{Y}} \bigg[\frac{1}{n} \sum_{i=1}^n
\big(\tilde{y}_i-\hf_n(x_i)\big)^2 -
\frac{1}{n} \sum_{i=1}^n \big(y_i-\hf_n(x_i)\big)^2\bigg] =
\frac{1}{n} \sum_{i=1}^n \mE_X
\Cov \big(y_i, \hf_n(x_i) \,|\, X\big).
$$
Some simple results immediately follow.

\begin{prop}\mbox{}
\label{prop:opts}
\begin{enumerate}
\item[(i)] If $T(X,Y)$ is an unbiased estimator of \OptF{} in
  the Fixed-X setting, for any $X$ in the support of $Q^n$, then it is
  also unbiased for \OptS{} in the Same-X setting.
\item[(ii)] If \OptF{} in the Fixed-X setting does not depend on $X$
  (e.g., as is true in least squares regression), then it is
  equal to \OptS{} in the Same-X setting.
\end{enumerate}
\end{prop}

\noindent
Some consequences of this proposition are as follows.
\begin{itemize}
\item For the least squares regression estimator of $Y$ on $X$
  (and $X$ having full column rank almost surely under $Q^n$), we have
  $\OptF=\OptS=2\sigma^2 p/ n$.
\item For a linear smoother, where \smash{$\hf_n(x_i)=s(x_i)^T Y$},
  $i=1,\ldots,n$ and we denote by $S(X) \in \mR^{n\times n}$ the
  matrix with rows $s(x_1),\ldots,s(x_n)$, we have (by direct
  calculation) $\OptF = 2\sigma^2\tr(S(X)) / n$ and
  $\OptS = 2\sigma^2\mE_X[\tr(S(X))]/n$.
\item For the lasso regression estimator of $Y$ on $X$ (and $X$ being
  in general position almost surely under $Q^n$), and a normal
  data model (i.e., the model in \eqref{eq:iid_pairs},
  \eqref{eq:eps_indep_x} with $\epsilon \sim N(0,\sigma^2)$),
  \citet{zou2007degrees,tibshirani2012degrees,tibshirani2013lasso}
  prove that for any value of the lasso tuning parameter
  $\lambda > 0$ and any $X$, the Fixed-X optimism is just
  \smash{$\OptF = 2\sigma^2 \mE_Y|A_\lambda(X,Y)| /n$}, where
  \smash{$A_\lambda(X,Y)$} is the active set at the lasso solution
  at $\lambda$ and \smash{$|A_\lambda(X,Y)|$} is its size;
  therefore we also have
  \smash{$\OptS = 2\sigma^2 \mE_{X,Y}|A_\lambda(X,Y)|/n$}.
\end{itemize}

Overall, we conclude that for the estimation of prediction error,
the Same-X setting is basically identical to Fixed-X.
We will see next that the situation is different for Random-X.

\subsection{Optimism for Random-X}

For the definition of Random-X optimism, we have to now integrate over
all sources of uncertainty,
$$
\OptR = \mE_{X,Y,x_0,y_0} \Big[\big(y_0-\hf_n(x_0)\big)^2 -
\big(y_1-\hf_n(x_1)\big)^2\Big].
$$
The definitions of $\OptS,\OptR$ are both given by a type of
prediction error (Same-X or Random-X) minus training error, and
there is just one common way to define training error.
Hence, by subtracting training error from both sides in the
decomposition \eqref{eq:er}, \eqref{eq:eb}, \eqref{eq:ev}, we obtain
the relationship:
\begin{equation}
\label{eq:optr}
\OptR = \OptS + \eb + \ev,
\end{equation}
where $\eb,\ev$ are the excess bias and variance as defined in
\eqref{eq:eb}, \eqref{eq:ev}, respectively.

As a consequence of our definitions, Random-X optimism is tied to
Same-X optimism by excess bias and variance terms, as in
\eqref{eq:optr}. The practical utility of this relationship: an
unbiased estimate of Same-X optimism (which, as pointed out
in the last subsection, follows straightforwardly from
an unbiased estimate of Fixed-X optimism), combined
with estimates of excess bias and variance, leads to an estimate
for Random-X prediction error.

\section{Excess bias and variance for least squares regression}
\label{sec:least_squares_eb_ev}

In this section, we examine the case when \smash{$\hf_n$} is defined
by least squares regression of $Y$ on $X$, where we assume $X$ has
full column rank (or, when viewed as random, has full column rank
almost surely under its marginal distribution $Q^n$).

\subsection{Nonnegativity of $\eb,\ev$}

Our first result
concerns the signs of \eb{} and \ev.

\begin{thm}
\label{thm:least_squares_nonneg}
For \smash{$\hf_n$} the least squares regression estimator, we have
both $\eb \geq 0$ and $\ev \geq 0$.
\end{thm}

\begin{proof}
We prove the result separately for \ev{} and \eb.

\smallskip\smallskip
\noindent {\it Nonnegativity of \ev.}
For a function $g : \mR^p \to \mR$, we will write
\smash{$g(X)=(g(x_1),\ldots,g(x_n))\in \mR^n$}, the
vector whose components are given by applying $g$ to the rows of
$X$. Letting $X_0 \in \mR^{n\times p}$ be a matrix of test covariate
values, whose rows are i.i.d.\ draws from $Q$, we note that excess
variance in \eqref{eq:ev} can be equivalently expressed as
$$
\ev = \mE_{X,X_0} \frac{1}{n}
\tr\big[ \Cov\big( \hf_n(X_0) \,|\, X,X_0 \big)\big] -
\mE_X \frac{1}{n}
\tr\big[ \Cov\big( \hf_n(X) \,|\, X \big)\big].
$$
Note that the second term here is just
\smash{$\mE_X [(\sigma^2/n)\tr(X(X^T X)^{-1}X^T)] = \sigma^2 p/n$}.
The first term is
\begin{align}
\nonumber
\frac{\sigma^2}{n} \tr\big(\mE_{X,X_0} \big[ (X^T X)^{-1}
X_0^TX_0 \big]\big) &=
\frac{\sigma^2}{n} \tr\big(\mE \big[ (X^T X)^{-1} \big]
\mE [ X_0^TX_0]\big) \\
\label{eq:integrated_var}
&= \frac{\sigma^2}{n} \tr\big(\mE \big[ (X^T X)^{-1} \big]
\mE [ X^TX]\big),
\end{align}
where in the first equality we used the independence of $X$ and $X_0$,
and in the second equality we used the identical distribution of $X$
and $X_0$. Now, by a result of \citet{groves1969note}, we know that
\smash{$\mE[(X^tX)^{-1}]- [\mE (X^tX)]^{-1}$} is positive
semidefinite. Thus we have
$$
\frac{\sigma^2}{n} \tr\big(\mE \big[ (X^T X)^{-1} \big]
\mE [X^TX]\big) \geq
\frac{\sigma^2}{n} \tr\big(\big[ \mE(X^T X) \big]^{-1}
\mE [X^TX]\big) = \frac{\sigma^2 p}{n}.
$$
This proves $\ev \geq 0$.

\smallskip\smallskip
\noindent {\it Nonnegativity of \eb.}
This result is actually a special case of Theorem
\ref{thm:ridge_eb_nonneg}, and its proof follows from the proof
of the latter.

\end{proof}

An immediate consequence of this, from the relationship between
Random-X and Same-X prediction error in \eqref{eq:er}, \eqref{eq:eb},
\eqref{eq:ev}, is the following.

\begin{cor}
For \smash{$\hf_n$} the least squares regression estimator, we have
$\ErrR \geq \ErrS$.
\end{cor}

This simple result, that the Random-X (out-of-sample) prediction error
is always larger than the Same-X (in-sample) prediction error for
least squares regression, is perhaps not suprising; however, we have
not been able to find it proven elsewhere in the literature at the
same level of generality.  We emphasize that our result only
assumes \eqref{eq:iid_pairs}, \eqref{eq:eps_indep_x}
and places no other assumptions on the distribution of errors,
distribution of covariates, or the form of $f(x)=\mE(y|x)$.

We also note that, while this relationship may seem obvious, it is in
fact not universal. Later in Section \ref{sec:ridge_ev_neg}, we show
that the excess variance \ev{} in heavily-regularized ridge regression
is guaranteed to be negative, and this can even lead to $\ErrR <
\ErrS$.

\subsection{Exact calculation of \ev{} for normal covariates}

Beyond the nonnegativity of $\eb,\ev$, it is actually easy to
quantify \ev{} exactly in the case that the covariates follow a normal
distribution.

\begin{thm}
\label{thm:least_squares_normal}
Assume that $Q = N(0,\Sigma)$, where $\Sigma \in \mR^{p\times p}$ is
invertible, and $p<n-1$.  Then for the least squares regression
estimator,
$$
\ev = \frac{\sigma^2 p}{n} \frac{p+1}{n-p-1}.
$$
\end{thm}

\begin{proof}
As the rows of $X$ are i.i.d.\ from $N(0,\Sigma)$, we have
$X^T X \sim W(\Sigma,n)$, which denotes a Wishart distribution
with $n$ degrees of freedom, and so $\mE(X^T X)=n\Sigma$.
Similarly, \smash{$(X^T X)^{-1} \sim W^{-1} (\Sigma^{-1},n)$},
denoting an inverse Wishart with $n$ degrees of freedom,
and hence \smash{$\mE[(X^T X)^{-1}]=\Sigma^{-1}/(n-p-1)$}. From the
arguments in the proof of Theorem \ref{thm:least_squares_nonneg},
$$
\ev = \frac{\sigma^2}{n} \tr\big(\mE \big[ (X^T X)^{-1} \big]
\mE [X^TX]\big) - \frac{\sigma^2 p}{n} =
\frac{\sigma^2}{n} \tr\bigg( I_{p\times p} \frac{n}{n-p-1} \bigg)
- \frac{\sigma^2 p}{n} = \frac{\sigma^2p}{n} \frac{p+1}{n-p-1},
$$
completing the proof.
\end{proof}

Interestingly, as we see, the excess variance
\ev{} does not depend on the covariance matrix $\Sigma$ in the case
of normal covariates.
Moreover, we stress that (as a consequence of our decomposition and
definition of $\eb,\ev$), the above calculation does not rely on
linearity of $f(x)=\mE(y|x)$.

When $f(x)$ is linear, i.e., the linear model is unbiased, it is not hard
to see that $\eb=0$, and the next result follows from \eqref{eq:optr}.

\begin{cor}
\label{cor:least_squares_normal}
Assume the conditions of Theorem \ref{thm:least_squares_normal}, and
further, assume that $f(x)=x^T \beta$, a linear function of $x$. Then
for the least squares regression estimator,
$$
\OptR = \OptS + \frac{\sigma^2 p}{n}
\frac{p+1}{n-p-1} = \frac{\sigma^2 p}{n}
\bigg(2+\frac{p+1}{n-p-1}\bigg)
$$
\end{cor}

For the unbiased case considered in Corollary
\ref{cor:least_squares_normal}, the same result can be found in
previous works, in particular in \citet{dicker2013optimal}, where it
is proven in the appendix. It is also similar to older results
from \citet{stein1960multiple,tukey1967discussion,
hocking1976analysis,thompson1978selection1,thompson1978selection2},
which assume the pair $(x,y)$ is jointly normal (and thus also assume
the linear model to be unbiased).
We return to these older classical results in the next section.  When
bias is present, our decomposition is required, so that the
appropriate result would still apply to \ev.

\subsection{Asymptotic calculation of \ev{} for nonnormal covariates}

Using standard results from random matrix theory, the result of
Theorem \ref{thm:least_squares_normal} can be generalized to an
asymptotic result over a wide class of distributions.\footnote{We
  thank Edgar Dobriban for help in formulating and proving this
  result.}

\begin{thm}
\label{thm:least_squares_asymp}
Assume that $x \sim Q$ is generated as follows: we draw $z \in \mR^p$,
having i.i.d.\ components $z_i \sim F$, $i=1,\ldots,p$, where $F$ is
any  distribution with zero mean and unit variance, and then set
$x=\Sigma^{1/2} z$, where $\Sigma \in \mR^{p \times p}$ is positive
definite and $\Sigma^{1/2}$ is its symmetric square root.
Consider an asymptotic setup where $p/n \to \gamma \in (0,1)$ as $n
\to \infty$. Then
$$
\ev \to \frac{\sigma^2 \gamma^2}{1-\gamma}
\quad \text{as $n \to \infty$}.
$$
\end{thm}

\begin{proof}
Denote by $X_n = Z_n \Sigma^{1/2}$ the training covariate
matrix, where $Z_n$ has rows $z_1,\ldots,z_n$, and we use subscripts
of $X_n,Z_n$ to denote the dependence on $n$ in our asymptotic
calculations below. Then as in the proof of Theorem
\ref{thm:least_squares_nonneg},
$$
\ev = \frac{\sigma^2}{n} \tr\big(\mE \big[ (X_n^T X_n)^{-1} \big]
\mE [X_n^TX_n]\big)
= \frac{\sigma^2}{n} \tr\big(\mE \big[ (Z_n^T Z_n)^{-1} \big]
\mE [Z_n^TZ_n]\big)
= \frac{\sigma^2}{n} \tr\big( n \mE \big[ (Z_n^T Z_n)^{-1} \big]
\big).
$$
The second equality used the relationship $X_n=Z_n\Sigma^{1/2}$, and
the third equality used the fact that the entries of $Z_n$ are i.i.d.\
with mean 0 and variance 1. This confirms that \ev{} does not
depend on the covariance matrix $\Sigma$.

Further, by the Marchenko-Pastur theorem, the distribution of
eigenvalues $\lambda_1,\ldots,\lambda_p$ of \smash{$Z_n^T Z_n/n$}
converges to a fixed law, independent of $F$; more precisely, the
random measure $\mu_n$, defined by
$$
\mu_n(A) = \frac{1}{p} \sum_{i=1}^p 1\{\lambda_i \in A\},
$$
converges weakly to the Marchenko-Pastor law $\mu$.  We note that
$\mu$ has density bounded away from zero when $\gamma<1$.  As the
eigenvalues of \smash{$n(Z_n^T
  Z_n)^{-1}$} are simply $1/\lambda_1,\ldots,1/\lambda_p$, we also have
that the random measure \smash{$\tilde\mu_n$}, defined by
$$
\tilde\mu_n(A) = \frac{1}{p} \sum_{i=1}^p 1\{1/\lambda_i \in A\},
$$
converges to a fixed law, call it \smash{$\tilde\mu$}. Denoting the
mean of \smash{$\tilde\mu$} by $m$, we now have
$$
\ev = \frac{\sigma^2}{n} \tr\big( n \mE \big[ (Z_n^T Z_n)^{-1} \big]
\big) = \frac{\sigma^2 p}{n} \mE\bigg[\frac{1}{p} \sum_{i=1}^p
\frac{1}{\lambda_i}\bigg] \to \sigma^2 \gamma m \quad \text{as $n \to
  \infty$}.
$$
As this same asymptotic limit, independent of $F$, must agree with
specific the case in which $F=N(0,1)$, we can conclude from Theorem
\ref{thm:least_squares_normal} that $m=\gamma/(1-\gamma)$, which
proves the result.
\end{proof}

The next result is stated for completeness.

\begin{cor}
\label{cor:least_squares_asymp}
Assume the conditions of Theorem \ref{thm:least_squares_asymp}, and
moreover, assume that the linear model is unbiased for $n$ large
enough.  Then
$$
\OptR \to \sigma^2 \gamma \frac{2-\gamma}{1-\gamma}
\quad \text{as $n \to \infty$}.
$$
\end{cor}

It should be noted that the requirement of Theorem
\ref{thm:least_squares_asymp} that the covariate vector $x$ be
expressible as $\Sigma^{1/2} z$ with the entries of $z$ i.i.d.\ is not
a minor one, and limits the set of covariate distributions for
which this result applies, as has been discussed in the literature on
random matrix theory (e.g., \citealt{elkaroui2009concentration}). In
particular, left multiplication by the square root matrix
$\Sigma^{1/2}$ performs a kind of averaging operation. Consequently,
the covariates $x$ can either
have long-tailed distributions, or have complex dependence structures,
but not both, since then the averaging will mitigate any long tail of
the distribution $F$. In our simulations in Section
\ref{sec:least_squares_sims}, we examine some settings that combine
both elements, and indeed the value of \ev{} in such settings can
deviate substantially from what this theory suggests.

\section{Covariance penalties for Random-X least squares}
\label{sec:least_squares_cp}

We maintain the setting of the last section, taking \smash{$\hf_n$} to
be the least squares regression estimator of $Y$ on $X$, where $X$ has
full column rank (almost surely under its marginal distribution $Q$).

\subsection{A Random-X version of Mallows' Cp}

Let us denote \smash{$\mathrm{RSS}=\|Y-\hf_n(X)\|_2^2$}, and recall
Mallows' Cp \citep{mallows1973comments}, which is defined as
$\mathrm{Cp}=\mathrm{RSS}/n+2\sigma^2p/n$.
The results in Theorems \ref{thm:least_squares_normal} and
\ref{thm:least_squares_asymp} lead us to define the following
generalized covariance penalty criterion we term \rcp:
$$
\rcp = \mathrm{Cp} + \ev = \frac{\mathrm{RSS}}{n} +
\frac{\sigma^2 p}{n} \bigg(2+\frac{p+1}{n-p-1}\bigg).
$$
An asymptotic approximation is given by
\smash{$\rcp \approx \mathrm{RSS}/n + \sigma^2 \gamma
  (2+\gamma/(1-\gamma))$}, in a problem scaling where
$p/n \to \gamma \in (0,1)$.

\rcp{} is an unbiased estimate of Random-X prediction error when the
linear model is unbiased and the covariates are normally distributed,
and an asymptotically unbiased estimate of Random-X prediction
error when the conditions of Theorem \ref{thm:least_squares_asymp}
hold. As we demonstrate below, it is also quite an effective
measure, in the sense that it has much lower variance (in the
appropriate settings for the covariate distributions) compared to
other almost-unbiased measures of Random-X prediction error, such as
OCV (ordinary leave-one-out cross-validation) and GCV (generalized
cross-validation).  However, in addition to the dependence on the
covariate distribution as in Theorems \ref{thm:least_squares_normal}
and \ref{thm:least_squares_asymp}, two other major
drawbacks to the use of \rcp{} in practice should be acknowledged.

\begin{enumerate}
\item[(i)] {\it The assumption that $\sigma^2$ is known.} This
  obviously affects the use of Cp in Fixed-X situations as well, as
  has been noted in the literature.
\item[(ii)] {\it The assumption of no bias.} It is critical to note
  here the difference from Fixed-X or Same-X situations, where \OptS{}
  (i.e., Cp) is independent of the bias in the model and must only
  correct for the ``overfitting'' incurred by model fitting. In
  contrast, in Random-X, the existence of \eb, which is a component of
  \OptR{} not captured by the training error, requires taking it into
  account in the penalty, if we hope to obtain low-bias estimates of
  prediction error. Moreover, it is often desirable to assume
  nothing about the form of the true model $f(x)=\mE(y|x)$,  hence it
  seems unlikely that theoretical considerations like those presented
  in Theorems \ref{thm:least_squares_normal} and
  \ref{thm:least_squares_asymp} can lead to estimates of \eb.
\end{enumerate}

We now propose enhancements that deal with each of these problems separately.

\subsection{Accounting for unknown $\sigma^2$ in unbiased least
  squares}

Here, we assume that the linear model is unbiased, $f(x)=x^T \beta$,
but the variance $\sigma^2$ of the noise in \eqref{eq:iid_pairs} is
unknown. In the Fixed-X setting, it is customary to replace $\sigma^2$
in covariance penalty approach like Cp with the unbiased estimate
\smash{$\hs^2 = \mathrm{RSS}/(n-p)$}.
An obvious choice is to also use \smash{$\hs^2$} in place of
$\sigma^2$ in \rcp, leading to a generalized covariance penalty
criterion we call \smash{\hrcp}:
$$
\hrcp = \frac{\mathrm{RSS}}{n} +
\frac{\hs^2 p}{n} \bigg(2+\frac{p+1}{n-p-1}\bigg) =
\frac{\mathrm{RSS}(n-1)}{(n-p)(n-p-1)}.
$$
An asymptotic approximation, under the scaling $p/n \to \gamma \in
(0,1)$, is \smash{$\hrcp \approx \mathrm{RSS}/(n (1-\gamma)^2)$}.

This penalty, as it turns out, is exactly equivalent to the Sp
criterion of \citet{tukey1967discussion,sclove1969criteria}; see also
\citet{stein1960multiple,hocking1976analysis,
thompson1978selection1,thompson1978selection2}.  These authors all
studied the case in which $(x,y)$ is jointly normal, and therefore the
linear model is assumed correct for the full model and any
submodel. The asymptotic approximation, on other hand, is equivalent
to the GCV (generalized cross-validation) criterion of
\citet{craven1978smoothing,golub1979generalized}, though the
motivation behind the derivation of GCV is somewhat different.

Comparing \hrcp{} to \rcp{} as a model evaluation criterion, we can
see the price of estimating $\sigma^2$ as opposed to knowing it, in
their asymptotic approximations. Their expectations are similar when
the linear model is true, but the variance of (the asymptotic form) of
\smash{\hrcp} is roughly $1/(1-\gamma)^4$ times larger than that of
(the asymptotic form) of \rcp.  So when, e.g., $\gamma = 0.5$, the
price of not knowing $\sigma^2$ translates roughly into a 16-fold
increase in the variance of the model evaluation metric. This is
clearly demonstrated in our simulation results in the next section.

\subsection{Accounting for bias and estimating \eb}

Next, we move to assuming nothing about the underlying regression
function $f(x)=\mE(y|x)$, and we examine methods that account for the
resulting bias \eb.
First we consider the behavior of \smash{\hrcp} (or
equivalently Sp) in the case that bias is present.  Though this
criterion was not designed to account for bias at all, we will see
it still performs an inherent bias correction. A straightforward
calculation shows that in this case
$$
\mE_{X,Y} \mathrm{RSS} = (n-p) \sigma^2 + nB,
$$
where recall \smash{$B=\mE_X \|\mE(\hf_n(X)\,|\,X)-f(X)\|^2/n$},
generally nonzero in the current setting, and thus
$$
\mE_{X,Y} \hrcp = \sigma^2 \frac{n-1}{n-p-1} + B
\frac{n(n-1)}{(n-p)(n-p-1)} \approx \frac{\sigma^2}{1-\gamma} +
\frac{B}{(1-\gamma)^2},
$$
the last step using an asymptotic approximation, under the scaling
$p/n \to \gamma \in (0,1)$.  Note that the second term on the
right-hand side above is the (rough) implicit estimate of integrated
Random-X bias used by \smash{\hrcp}, which is larger than the
integrated Same-X bias $B$ by a factor of $1/(1-\gamma)^2$. Put
differently, \smash{\hrcp} implicitly assumes that \eb{} is (roughly)
$1/(1-\gamma)^2-1$ times as big as the Same-X
bias.  We see no reason to believe that this relationship
(between Random-X and Same-X biases) is generally correct, but it is
not totally naive either, as we will see empirically that
\smash{\hrcp} still provides reasonably good estimates of Random-X
prediction error in biased situations in Section
\ref{sec:least_squares_sims}.   A partial explanation is available
through a connection to OCV, as
discussed, e.g., in the derivation of GCV in
\citet{craven1978smoothing}. We return to this issue in Section
\ref{sec:discussion}.

We describe a more principled approach to estimating the
integrated Random-X bias, $B+\eb$, assuming knowledge of $\sigma^2$,
and leveraging a bias estimate implicit to OCV.
Recall that OCV builds $n$ models, each time leaving one observation
out, applying the fitted model to that observation, and using these
$n$ holdout predictions to estimate prediction error. Thus it gives
us an almost-unbiased estimate of Random-X prediction error $\ErrR$
(``almost'', because its training sets are all of size $n-1$ rather than
$n$). For least squares regression (and other estimators), the
well-known ``shortcut-trick'' for OCV
(e.g., \citealt{wahba1990spline,hastie2009elements}) allows us to
represent the OCV residuals in terms of  weighted training residuals.
Write \smash{$\hf_n^{(-i)}$} for the least squares estimator trained
on all but $(x_i,y_i)$, and $h_{ii}$ the $i$th diagonal element of
$X(X^T X)^{-1}X^T$, for $i=1,\ldots,n$. Then this trick tells
us that
$$
y_i - \hf_n^{(-i)}(x_i) = \frac{y_i-\hf_n(x_i)}{1-h_{ii}},
$$
which can be checked by applying the Sherman-Morrison update
formula for relating the inverse of a matrix to the inverse of its
rank-one pertubation.  Hence the OCV error can be expressed as
$$
\mathrm{OCV} = \frac{1}{n} \sum_{i=1}^n
\Big(y_i-\hf_n^{(-i)}(x_i)\Big)^2 = \frac{1}{n} \sum_{i=1}^n
\bigg(\frac{y_i-\hf_n(x_i)}{1-h_{ii}}\bigg)^2.
$$
Taking an expectation conditional on $X$, we find that
\begin{align}
\nonumber
\mE(\mathrm{OCV} | X) &=
\frac{1}{n} \sum_{i=1}^n \frac{\Var(y_i-\hf_n(x_i) \,|\,
  X)}{(1-h_{ii})^2} + \frac{1}{n} \sum_{i=1}^n
  \frac{[f(x_i)-\mE(\hf_n(x_i) \,|\, X)]^2}{(1-h_{ii})^2} \\
\label{eq:ocv_decomp}
&=\frac{\sigma^2}{n} \sum_{i=1}^n \frac{1}{1-h_{ii}} +
\frac{1}{n} \sum_{i=1}^n \frac{[f(x_i)-\mE(\hf_n(x_i) \,|\,
  X)]^2}{(1-h_{ii})^2},
\end{align}
where the second line uses \smash{$\Var(y_i-\hf_n(x_i)
\,|\, X)=(1-h_{ii})\sigma^2$}, $i=1,\ldots,n$.
The above display shows
$$
\mathrm{OCV}-\frac{\sigma^2}{n} \sum_{i=1}^n \frac{1}{1-h_{ii}} =
\frac{1}{n} \sum_{i=1}^n \Big( \big(y_i-\hf_n(x_i)\big)^2 - (1-h_{ii})
\sigma^2\Big) \frac{1}{(1-h_{ii})^2}
$$
is an almost-unbiased estimate of the integrated Random-X prediction
bias, $B+\eb$ (it is almost-unbiased, due to the almost-unbiased
status of OCV as an estimate of Random-X prediction error).
Meanwhile, an unbiased estimate of the integrated Same-X prediction
bias $B$ is
$$
\frac{\mathrm{RSS}}{n} - \frac{\sigma^2 (n-p)}{n} = \frac{1}{n}
\sum_{i=1}^n \Big(\big(y_i-\hf_n(x_i)\big)^2 - (1-h_{ii}) \sigma^2 \Big).
$$
Subtracting the last display from the second to last delivers
$$
\widehat\eb = \frac{1}{n} \sum_{i=1}^n \Big(\big(y_i-\hf_n(x_i)\big)^2 -
(1-h_{ii}) \sigma^2 \Big) \bigg(\frac{1}{(1-h_{ii})^2} - 1\bigg),
$$
an almost-unbiased estimate of the excess bias \eb.   We now define a
generalized covariance penalty criterion that we call \rcpp{} by
adding this to \rcp:
$$
\rcpp = \mathrm{RCp} + \widehat\eb
= \mathrm{OCV} - \frac{\sigma^2}{n} \sum_{i=1}^n
\frac{h_{ii}}{1-h_{ii}}  + \frac{\sigma^2 p}{n}
\bigg(1+\frac{p+1}{n-p-1}\bigg).
$$
It is worth pointing out that, like \rcp{} and \smash{\hrcp{}}, \rcpp
assumes that we are in a setting covered by Theorem
\ref{thm:least_squares_normal} or asymptotically by Theorem
\ref{thm:least_squares_asymp}, as it takes advantage of the value of
\ev{} prescribed by these theorems.

A key question, of course, is: what have we achieved by moving from
OCV to \rcpp, i.e., can we explicitly show that \rcpp{} is preferable to
OCV for estimating Random-X prediction error when its assumptions hold?
We give a partial positive answer next.

\subsection{Comparing \rcpp and OCV}
\label{sec:compare_rcpp_ocv}

As already discussed, OCV is by an almost-unbiased estimate of
Random-X prediction error (or an unbiased estimate of Random-X
prediction error for the procedure in question, here least squares,
applied to a training set of size $n-1$).  The decomposition in
\eqref{eq:ocv_decomp} demonstrates its variance and bias components,
respectively, conditional on $X$.  It should be emphasized that OCV
has the significant advantage over \rcpp{} of not requiring
knowledge of $\sigma^2$ or assumptions on $Q$. Assuming that
$\sigma^2$ is known and $Q$ is well-behaved,
we can compare the two criteria for estimating Random-X prediction error in
least squares.

OCV is generally slightly conservative as an
estimate of Random-X prediction error, as models trained on
more observations are generally expected to be better.  \rcpp{} does
not suffer
from such slight conservativeness in the variance component, relying
on the integrated variance from theory, and in that regard
it may already be seen as an improvement.  However we will choose to
ignore this issue of conservativeness, as the difference in training
on $n-1$ versus $n$ observations is clearly small when $n$ is large.
Thus, we can approximate the mean squared error or MSE of each
method, as an estimate of Random-X prediction error, as
\begin{align*}
\mE(\mathrm{OCV}-\ErrR)^2 &\approx
\Var_X \big(\mE(\mathrm{OCV}|X)\big) +
\mE_X \big(\Var(\mathrm{OCV}|X)\big),  \\
\mE(\rcpp - \ErrR)^2 &\approx
\Var_X \big(\mE(\rcpp|X)\big) +
\mE_X \big(\Var(\rcpp|X)\big),
\end{align*}
where these two approximations would be equalities if OCV and \rcpp{}
were exactly unbiased estimates of $\ErrR$.  Note that
conditioned on $X$, the difference between OCV and \rcpp{},
is nonrandom (conditioned on $X$, all
diagonal entries $h_{ii}$, $i=1,\ldots,p$ are nonrandom).  Hence
$\mE_X\Var(\mathrm{OCV}|X)=\mE_X\Var(\rcpp|X)$, and we are
left to compare $\Var_X\mE(\mathrm{OCV}|X)$ and
$\Var_X\mE(\rcpp|X)$, according to the (approximate) expansions above,
to compare the MSEs of OCV and \rcpp{}.

Denote the two terms in \eqref{eq:ocv_decomp} by $v(X)$ and $b(X)$,
respectively, so that $\mE(\mathrm{OCV}|X)=v(X)+b(X)$ can be viewed as
a decomposition into variance and bias components, and note that by
construction
$$
\Var_X \big(\mE(\mathrm{OCV}|X)\big) = \Var_X\big(v(X)+b(X)\big)
\quad\text{and}\quad
\Var_X \big(\mE(\rcpp|X)\big) = \Var_X\big(b(X)\big).
$$
It seems reasonable to believe that \smash{$\Var_X\mE(\mathrm{OCV}|X)
  \geq \Var_X\mE(\rcpp|X)$} would hold in most cases, thus
\rcpp{} would be no worse than OCV. One situation in which this occurs
is the case
when the linear model is unbiased, hence $b(X)=0$ and consequently
\smash{$\Var_X \big(\mE(\rcpp|X)\big) = \Var_X\big(b(x)\big) = 0$}.
In general, \smash{$\Var_X\mE(\mathrm{OCV}|X)
  \geq \Var_X\mE(\rcpp|X)$} is guaranteed when
\smash{$\Cov_X(v(X),b(X)) \geq 0$}. This means that
choices of $X$ that give large variance tend to also give large
bias, which seems reasonable to assume and indeed appears to be
true in our experiments.  But,
this covariance depends on the underlying mean function
$f(x)=\mE(y|x)$ in complicated ways, and at the moment it
eludes rigorous analysis.


\section{Simulations for least squares regression}
\label{sec:least_squares_sims}

We empirically study the decomposition of Random-X prediction
error into its various components for least squares regression in
different problem settings, and examine the performance of the
various model evaluation criteria in these settings. The only
criterion which is assumption-free and should invariably give unbiased
estimates of Random-X prediction error is OCV (modulo the slight bias
in using $n-1$ rather than $n$ training
observations). Thus we may consider OCV as the ``gold standard''
approach, and we will hold the other methods up to its standard under
different conditions, either when the assumptions they use hold or are
violated.

Before diving into the details, here is a high-level summary of the
results: \rcp{} performs very well in unbiased settings (when the mean
is linear), but very poorly in biased ones (when the mean is
nonlinear); \rcpp{} and \smash{\hrcp{}} peform well overall, with
\smash{\hrcp{}} having an advantage and even holding a small
advantage over OCV, in essentially all settings, unbiased and biased.
This is perhaps a bit surprising since \smash{\hrcp{}} is not
designed to account for bias, but then again, not as surprising once
we recall that \smash{\hrcp{}} is closely related to GCV.

We perform experiments in a total of six data generating mechanisms,
based on three different distributions $Q$ for the covariate vector
$x$, and two models for $f(x)=\mE(y|x)$, one unbiased (linear) and the
other biased (nonlinear). The three generating models for $x$ are
as follows.
 \begin{itemize}
\item {\it Normal.}  We choose $Q=N(0,\Sigma)$, where $\Sigma$ is
  block-diagonal, containing five blocks such that all variables in a
  block have pairwise correlation 0.9.
\item {\it Uniform.} We define $Q$ by taking $N(0,\Sigma)$ as above,
  then apply the inverse normal distribution function componentwise.
  In other words, this can be seen as a Gaussian copula with
  uniform marginals.
\item {\it $t(4)$.} We define $Q$ by taking $N(0,\Sigma)$ as above,
  then adjust the marginal distributions appropriately, again a
  Gaussian copula with $t(4)$ marginals.
\end{itemize}
Note that Theorem \ref{thm:least_squares_normal} covers the
normal setting (and in fact, the covariance matrix $\Sigma$ plays no
role in the \rcp{} estimate), while the uniform and $t(4)$ settings do
not comply with either Theorems \ref{thm:least_squares_normal} or
\ref{thm:least_squares_asymp}.  Also, the latter two settings differ
considerably in the nature of the distribution $Q$: finite support
versus long tails, respectively. The two generating models for $y|x$
both use $\epsilon \sim N(0,20^2)$, but differ in the specification
for the mean function $f(x)=\mE(y|x)$, as follows.
 \begin{itemize}
\item {\it Unbiased.} We set $f(x)=\sum_{j=1}^p  x_j$.
\item {\it Biased.} We set $f(x) = C \sum_{j=1}^p |x_j|$.
\end{itemize}

The simulations discussed in the coming subsections all use $n=100$
training observations. In the ``high-dimensional'' case, we use $p=50$
variables and $C=0.75$, while in the ``low-dimensional, extreme bias''
case, we use $p=10$ and $C=100$. In both cases, we use a test set
of $10^4$ observations to evaluate Random-X quantities like
$\ErrR,\eb,\ev$. Lastly, all figures show results averaged over 5000
repetitions.

\subsection{The components of Random-X prediction error}

We empirically evaluate $B,V,\eb,\ev$ for least squares regression
fitted in the six settings (three for the distribution of
$x$ times two for $\mE(y|x)$) in the high-dimensional
case, with $n=100$ and $p=50$.  The results are shown
in Figure \ref{fig:decomp}.  We can see
the value of \ev{} implied by Theorem
\ref{thm:least_squares_normal} is extremely accurate for the
normal setting, and also very accurate for the short-tailed
uniform setting. However for the $t(4)$ setting, the value of
\ev{} is quite a bit higher than what the theory implies. In
terms of bias, we observe that for the biased settings the value of
\eb{} is bigger than the Same-X bias $B$, and so it must be taken into
account if we hope to obtain reasonable estimates of Random-X
prediction error $\ErrR$.

\begin{figure}[htb]
\centering
\includegraphics[width=0.7\textwidth]{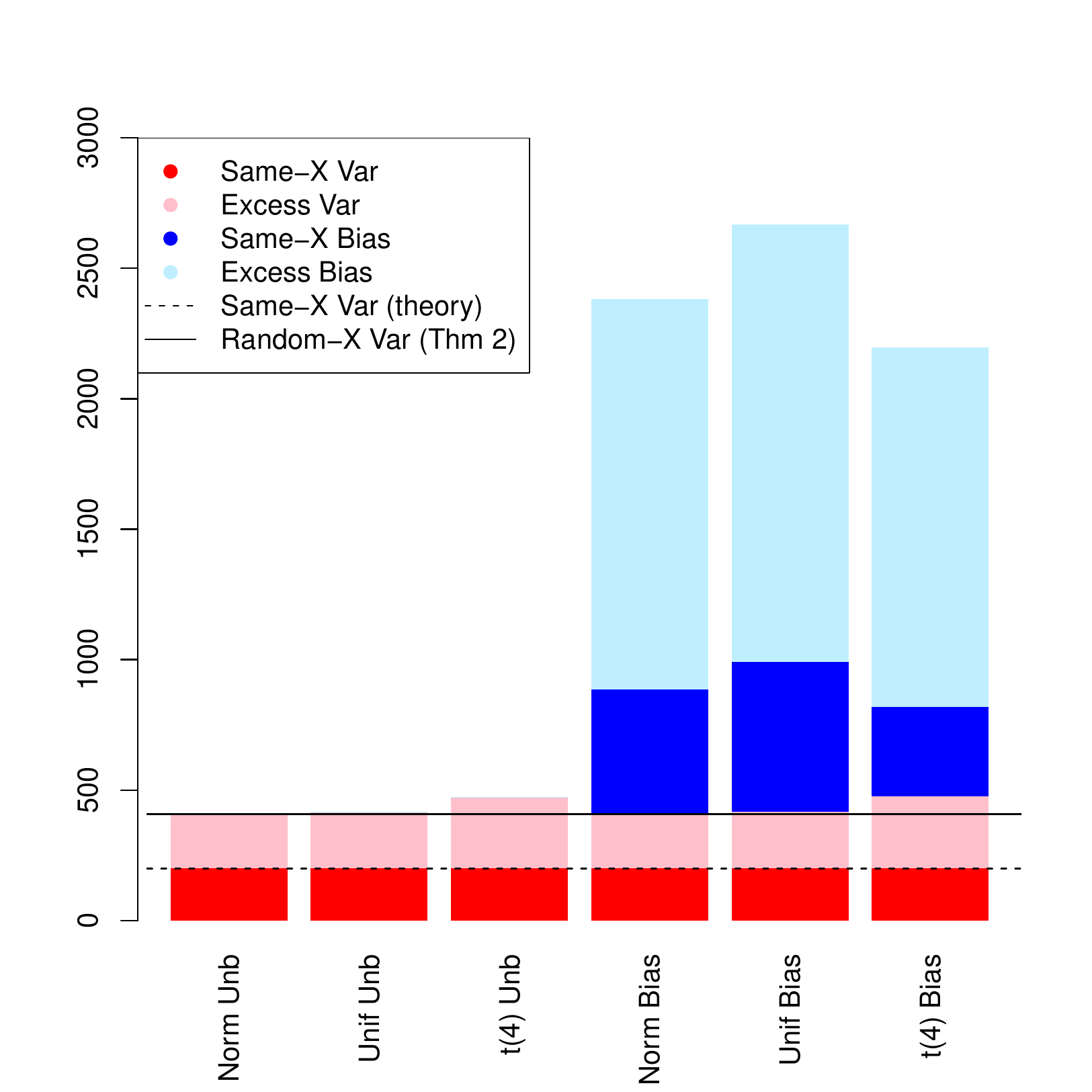}
\caption{\it Decomposition of Random-X prediction error into its
  reducible components: Same-X bias $B$, Same-X variance $V$,
  excess bias \eb, and excess variance \ev, in the
  ``high-dimensional'' case with $n=100$ and $p=50$.}
\label{fig:decomp}
\end{figure}

\subsection{Comparison of performances in estimating prediction error}

Next we compare the performance of the proposed criteria for
estimating the Random-X prediction error of least squares over the six
simulation settings.   The results in Figures \ref{fig:relative_high} and
\ref{fig:relative_low} correspond to the
``high-dimensional'' case with $n=100$ and $p=50$ and the
``low-dimensional, extreme bias'' case with $n=100$ and $p=10$,
respectively.  Displayed are the MSEs in estimating the Random-X
prediction error, relative to OCV; also, the MSE for each
method are broken down into squared bias and variance components.

\begin{figure}[htb]
\centering
\includegraphics[width=0.75\textwidth]{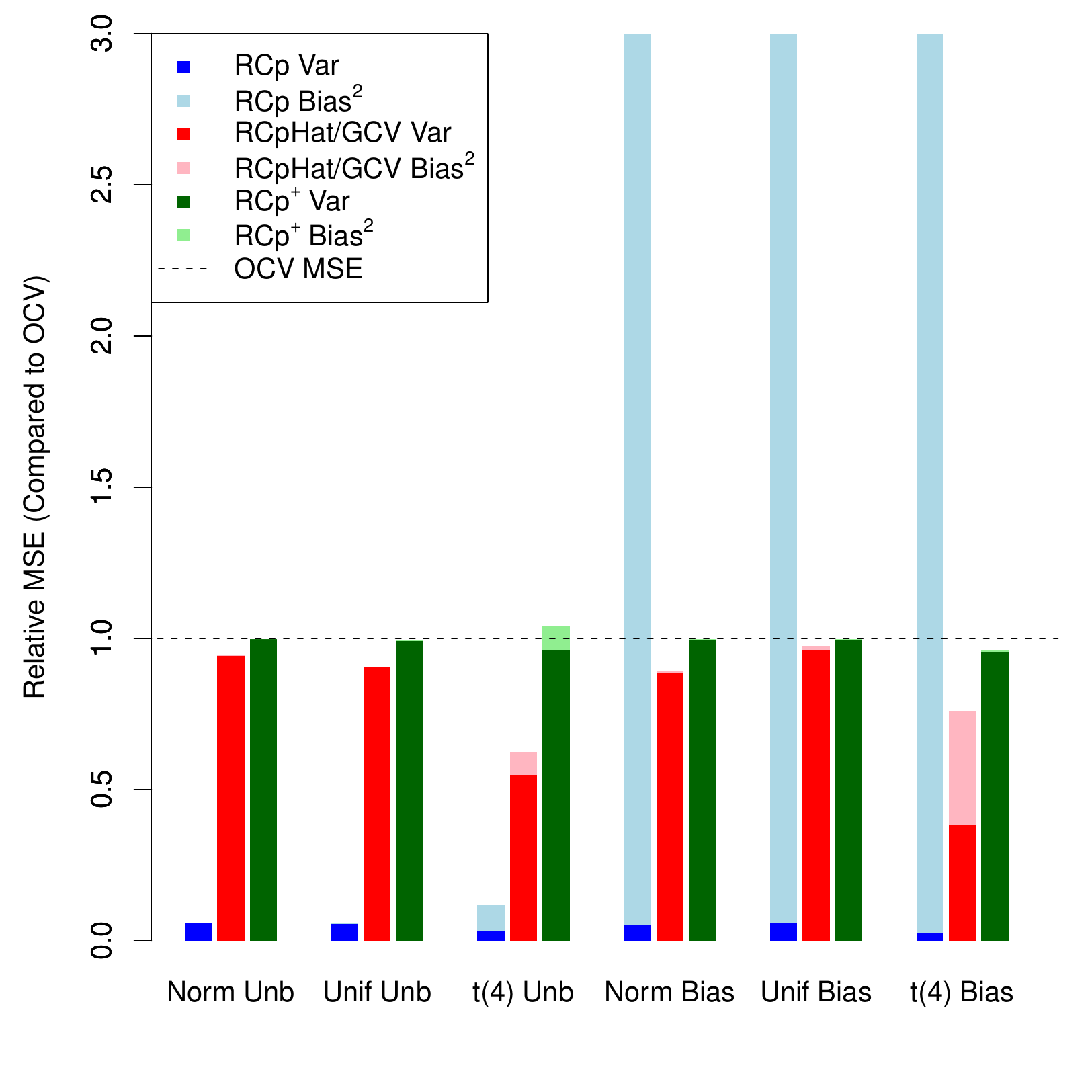}
\caption{\it MSEs of the different methods in estimating
  Random-X prediction error relative to OCV, in the
  ``high-dimensional'' case with $n=100$ and $p=50$.  The
  plot is truncated at a relative error of 3 for clarity, but the
  \rcp{} relative errors continue as high as 10 in the biased
  settings.}
\label{fig:relative_high}
\end{figure}

In the high-dimensional case in Figure \ref{fig:relative_high}, we
see that for the true linear models (three leftmost scenarios),
\rcp{} has by  far the lowest MSE in estimating Random-X prediction
error, much better than OCV. For the normal and uniform
covariate distributions, it also has no bias in estimating this error,
as warranted by Theorem \ref{thm:least_squares_normal} for the normal
setting.  For the $t(4)$ distribution, there is
already significant bias in the prediction error estimates generated
by \rcp, as is expected from the results in Figure \ref{fig:decomp};
however, if the linear model is correct then we see \rcp{} still has
three- to five-fold lower MSE compared to all other methods.
The situation changes dramatically when bias is added (three rightmost
scenarios). Now, \rcp{} is by far the worse method,
failing completely to account for large \eb, and its relative MSE
compared to OCV reaches as high as 10.

As for \rcpp{} and \smash{\hrcp{}} in the high-dimensional case, we
see that \rcpp{} indeed has lower error than OCV under the normal
models as argued in Section \ref{sec:compare_rcpp_ocv}, and also in
the uniform models. This is true regardless of the presence of bias.
The difference, however is small: between 0.1\% and 0.7\%. In these
settings, we can see \smash{\hrcp{}} has even lower MSE than \rcpp,
with no evident bias in dealing with the biased models.
For the long-tailed $t(4)$ distribution, both \smash{\hrcp{}} and
\rcpp{} suffer some bias in estimating prediction error, as
expected. Interestingly, in the nonlinear model with $t(4)$
covariates (rightmost scenario), \smash{\hrcp{}} does suffer
significant bias in estimating prediction error, as opposed to
\rcpp. However, this bias does not offset the increased variance due
to \rcpp/OCV.

\begin{figure}[htb]
\centering
\includegraphics[width=0.75\textwidth]{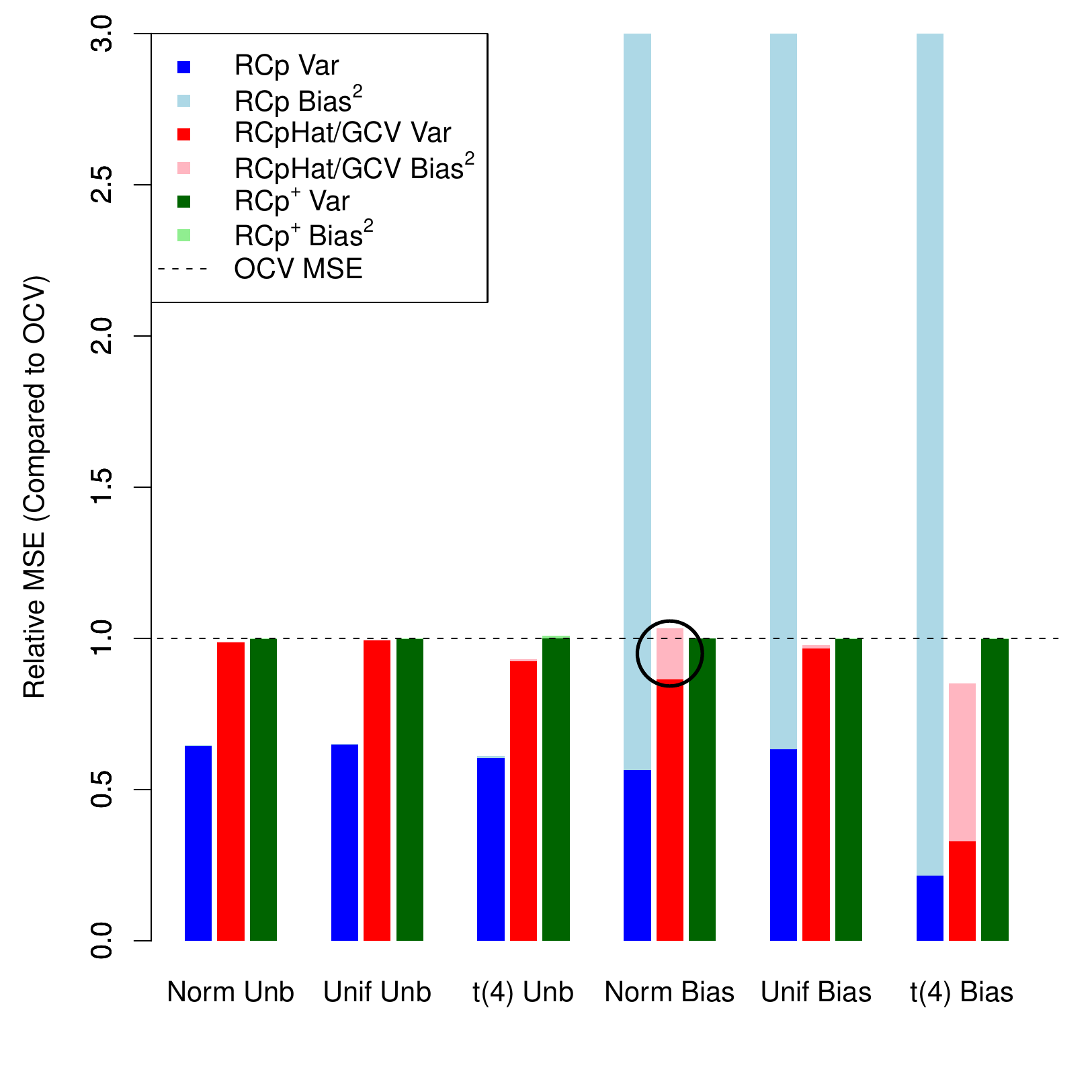}
\caption{\it MSEs of the different methods in estimating
  Random-X prediction error relative to OCV, in the
  ``low-dimensional, extreme bias'' case with $n=100$ and $p=10$.}
\label{fig:relative_low}
\end{figure}

In the low-dimensional case in Figure \ref{fig:relative_low}, many of
the same conclusions apply: \rcp{} does well if the linear
model is correct, even with the long-tailed covariate distribution,
but fails completely in the presence of nonlinearity. Also, \rcpp{}
performs almost identically to OCV throughout. The most important
distinction is the failure of \smash{\hrcp{}} in the normal covariate,
biased setting, where it suffers significant bias in estimating the
prediction error (see circled region in the plot). This demonstrates
that the heuristic correction for \eb employed by \smash{\hrcp{}} can
fail when the linear model does not hold, as opposed to \rcpp{} and
OCV. We discuss this further in Section \ref{sec:discussion}.


\section{The effects of ridge regularization}
\label{sec:ridge}

In this section, we examine ridge regression, which behaves
similarly in some ways to least squares regression, and differently in
others. In particular, like least squares, it has nonnegative excess
bias, but unlike least squares, it can have negative excess variance,
increasingly so for larger amounts of regularization.

These results are established in the subsections below, where we study
excess bias and variance separately.  Throughout, we will write
\smash{$\hf_n$} for the estimator from the ridge regression of $Y$ on
$X$, i.e.,
\smash{$\hf_n(x)=x^T (X^T X + \lambda
  I)^{-1} X^T Y$}, where the tuning parameter $\lambda > 0$ is
considered arbitrary (and for simplicity, we make the dependence of
\smash{$\hf_n$} on $\lambda$ implicit).
When $\lambda=0$, we must assume that $X$ has full column rank
(almost surely under its marginal distribution $Q^n$), but when
$\lambda>0$, no assumption is needed on $X$.

\subsection{Nonnegativity of \eb}

We prove an extension to the excess bias result in Theorem
\ref{thm:least_squares_nonneg} for least squares regression that the
excess bias in ridge regression is nonnegative.

\begin{thm}
\label{thm:ridge_eb_nonneg}
For \smash{$\hf_n$} the ridge regression estimator, we have $\eb \geq
0$.
\end{thm}

\begin{proof}
This result is actually itself a special case of Theorem
\ref{thm:rkhs_eb_nonneg}; the latter is phrased in somewhat of
a different (functional) notation, so for concreteness, we give a
direct proof of the result for ridge regression here.
Let $X_0 \in \mR^{n\times p}$ be a matrix of test covariate values,
with rows i.i.d.\ from $Q$, and let $Y_0 \in \mR^n$ be a vector of
associated test response value. Then excess bias in \eqref{eq:eb} can
be written as
$$
\eb = \mE_{X,X_0} \frac{1}{n} \big\|
\mE\big(\hf_n(X_0) \,|\, X,X_0 \big)-f(X_0)\big\|_2^2 -
\mE_X \frac{1}{n} \big\|
\mE\big(\hf_n(X) \,|\, X \big)-f(X)\big\|_2^2.
$$
Note \smash{$\hf_n(X)=X(X^T X+\lambda I)^{-1} X^T Y$}, and by
linearity, \smash{$\mE(\hf_n(X) \,|\, X)=X(X^T X+\lambda I)^{-1} X^T
  f(X)$}.
Recalling the optimization problem underlying ridge regression, we
thus have
$$
\mE\big(\hf_n(X) \,|\, X\big)= \argmin_{X\beta \in \mR^n} \;
\| f(X) - X\beta \|_2^2 + \lambda\|\beta\|_2^2.
$$
An analogous statement holds for \smash{$\hf_{0n}$}, which
we write to denote the result from the ridge regression $Y_0$ on
$X_0$; we have
$$
\mE\big(\hf_{0n}(X_0) \,|\, X_0\big)= \argmin_{X_0\beta \in \mR^n} \;
\| f(X_0) - X_0\beta \|_2^2 + \lambda\|\beta\|_2^2.
$$
Now write \smash{$\beta_n=(X^T X + \lambda I)^{-1} X^T f(X)$} and
\smash{$\beta_{0n}=(X_0^T X_0 + \lambda I)^{-1} X_0^T f(X_0)$} for
convenience. By optimality of \smash{$X_0\beta_{0n}$} for
the minimization problem in the last display,
$$
\| X_0 \beta_n - f(X_0)\|_2^2 + \lambda \|\beta_n\|_2^2 \geq
\| X_0 \beta_{0n} - f(X_0)\|_2^2 + \lambda \|\beta_{0n}\|_2^2,
$$
and taking an expectation over $X,X_0$ gives
\begin{align*}
\mE_{X,X_0} \Big[\| X_0 \beta_n - f(X_0)\|_2^2 + \lambda
\|\beta_n\|_2^2\Big] &\geq
\mE_{X_0} \Big[\| X_0 \beta_{0n} - f(X_0)\|_2^2 + \lambda
\|\beta_{0n}\|_2^2\Big] \\
&=
\mE_X \Big[\| X \beta_n - f(X)\|_2^2 + \lambda
\|\beta_n\|_2^2\Big],
\end{align*}
where in the last line we used the fact that $(X,Y)$ and $(X_0,Y_0)$
are identical in distribution.  Cancelling out the common term of
\smash{$\lambda\, \mE_X \|\beta_n\|_2^2$} in the first and third lines
above establishes the result, since
\smash{$\mE(\hf_n(X_0) \,|\, X,X_0) = X_0\beta_n$} and
\smash{$\mE(\hf_n(X) \,|\, X)=X\beta_n$}.
\end{proof}

\subsection{Negativity of \ev{} for large $\lambda$}
\label{sec:ridge_ev_neg}

Here we present two complementary results on the variance side.

\begin{prop}
\label{prop:ridge_var_decr}
For \smash{$\hf_n$} the ridge regression estimator, the integrated
Random-X prediction variance,
$$
V+\ev = \mE_{X,x_0} \Var\big(\hf_n(x_0) \,|\, X,x_0\big),
$$
is a nonincreasing function of $\lambda$.
\end{prop}

\begin{proof}
As in the proofs of Theorems \ref{thm:least_squares_nonneg} and
\ref{thm:ridge_eb_nonneg}, let
$X_0 \in \mR^{n\times p}$ be a test covariate matrix, and
notice that we can write the integrated Random-X variance as
$$
V+\ev = \mE_{X,X_0} \frac{1}{n}
\tr\big[ \Cov\big( \hf_n(X_0) \,|\, X,X_0 \big)\big].
$$
For a given value of $X,X_0$, we have
\begin{align*}
\frac{1}{n} \tr\big[ \Cov\big( \hf_n(X_0) \,|\, X,X_0 \big)\big] &=
\frac{\sigma^2}{n} \tr \Big( X_0(X^T X + \lambda I)^{-1} X^T X
(X^T X + \lambda I)^{-1} X_0^T \Big) \\
&= \frac{\sigma^2}{n} \tr \bigg( X_0^T X_0 \sum_{i=1}^p u_iu_i^T
\frac{d_i^2}{(d_i^2 + \lambda)^2} \bigg),
\end{align*}
where the second line uses an eigendecomposition $X^T X =
U D U^T$, with $U \in \mR^{p\times p}$ having orthonormal columns
$u_1,\ldots,u_p$ and
$D=\mathrm{diag}(d_1^2,\ldots,d_p^2)$. Taking a derivative with
respect to $\lambda$, we see
$$
\frac{d}{d\lambda} \Bigg(
\frac{1}{n} \tr\big[ \Cov\big( \hf_n(X_0) \,|\, X,X_0 \big)\big]\Bigg)
= -2\frac{\sigma^2}{n} \tr \bigg( X_0^T X_0 \sum_{i=1}^p u_iu_i^T
\frac{\lambda d_i^2}{(d_i^2 + \lambda)^3} \bigg) \leq 0,
$$
the inequality due to the fact that $\tr(AB) \geq 0$ if $A,B$ are
positive semidefinite matrices. Taking an expectation and
switching the order of integration and differentiation (which is
possible because the integrand is a continuously differentiable
function of $\lambda>0$) gives
$$
\frac{d}{d\lambda} \Bigg(
\mE_{X,X_0} \frac{1}{n} \tr\big[ \Cov\big( \hf_n(X_0) \,|\, X,X_0
\big)\big]\Bigg) =
\mE_{X,X_0} \frac{d}{d\lambda} \Bigg(
\frac{1}{n} \tr\big[ \Cov\big( \hf_n(X_0) \,|\, X,X_0
\big)\big]\Bigg) \leq 0,
$$
the desired result.
\end{proof}

The proposition shows that adding regularization guarantees a decrease
in variance for Random-X prediction.  The same is true of the variance
in Same-X prediction.  However, as we show next, as the amount
of regularization increases these two variances decrease at different
rates, a phenomenon that manifests itself in the fact that and the
Random-X prediction variance is guaranteed to be smaller than the
Same-X prediction variance for large enough $\lambda$.

\begin{thm}
\label{thm:ridge_ev_neg}
For \smash{$\hf_n$} the ridge regression estimator, the integrated
Same-X prediction variance and integrated Random-X prediction variance
both approach zero as $\lambda \to \infty$.  Moreover, the limit of
their ratio satisfies
$$ \lim_{\lambda\to \infty}
\frac{\mE_{X,x_0} \Var(\hf_n(x_0) \,|\, X,x_0)}
{\mE_{X} \Var(\hf_n(x_1) \,|\, X)}
= \frac{ \tr[\mE(X^TX) \mE(X^TX)]}
{\tr[\mE(X^TX X^TX)]} \leq 1,
$$
the last inequality reducing to an equality if and only if $x
\sim Q$ is deterministic and has no variance.
\end{thm}

\begin{proof}
Again, as in the proof of the last proposition as well as Theorems
\ref{thm:least_squares_nonneg} and \ref{thm:ridge_eb_nonneg},
let $X_0 \in \mR^{n\times p}$ be a test covariate matrix, and
write the integrated Same-X and Random-X prediction variances as
$$
\mE_X \frac{1}{n} \tr\big[ \Cov\big(\hf_n(X) \,|\, X\big)\big]
\quad \text{and} \quad
\mE_{X,X_0} \frac{1}{n} \tr\big[ \Cov\big(\hf_n(X_0) \,|\, X,X_0
\big)\big],
$$
respectively.  From the arguments in the proof of Proposition
\ref{prop:ridge_var_decr}, letting $X^T X =
U D U^T$ be an eigendecomposition with $U \in \mR^{p\times p}$
having orthonormal columns $u_1,\ldots,u_p$ and
$D=\mathrm{diag}(d_1^2,\ldots,d_p^2)$, we have
\begin{align*}
\lim_{\lambda \to \infty} \mE_{X,X_0} \frac{1}{n} \tr\big[
  \Cov\big(\hf_n(X_0) \,|\, X,X_0 \big)\big]
&= \lim_{\lambda \to \infty}\mE_{X,X_0}
\frac{\sigma^2}{n} \tr \bigg( X_0^T X_0 \sum_{i=1}^p u_iu_i^T
\frac{d_i^2}{(d_i^2 + \lambda)^2} \bigg) \\
&= \mE_{X,X_0} \lim_{\lambda \to \infty}
\frac{\sigma^2}{n} \tr \bigg( X_0^T X_0 \sum_{i=1}^p u_iu_i^T
\frac{d_i^2}{(d_i^2 + \lambda)^2} \bigg) = 0,
\end{align*}
where in the second line we used the dominated convergence theorem to
exchange the limit and the expectation (since \smash{$\mE_{X,X_0} \tr
  [X_0^T X_0 \sum_{i=1}^p u_iu_i^T (d_i^2/(d_i^2+\lambda)^2)] \leq
  \mE_{X,X_0}\tr(X_0^T X_0 X^T X) < \infty$}).
Similar arguments show that the integrated Same-X prediction variance
also tends to zero.

Now we consider the limiting ratio of the integrated variances,
\begin{align*}
\lim_{\lambda \to \infty}
\frac{\mE_{X,X_0} \tr[ \Cov(\hf_n(X_0) \,|\, X,X_0)]}
{\mE_X \tr[ \Cov(\hf_n(X) \,|\, X)]} &=
\lim_{\lambda \to \infty} \frac{\mE_{X,X_0} \tr[
X_0(X^T X + \lambda I)^{-1} X^T X
(X^T X + \lambda I)^{-1} X_0^T]}
{\mE_X \tr[X(X^T X + \lambda I)^{-1} X^T X
(X^T X + \lambda I)^{-1} X^T]} \\
&= \lim_{\lambda \to \infty} \frac{\mE_{X,X_0} \tr[
\lambda^2 X_0(X^T X + \lambda I)^{-1} X^T X
(X^T X + \lambda I)^{-1} X_0^T]}
{\mE_X \tr[\lambda^2 X(X^T X + \lambda I)^{-1}
X^T X (X^T X + \lambda I)^{-1} X^T]} \\
&= \frac{\lim_{\lambda \to \infty} \mE_{X,X_0} \tr[
\lambda^2 X_0(X^T X + \lambda I)^{-1} X^T X
(X^T X + \lambda I)^{-1} X_0^T]}
{\lim_{\lambda \to \infty} \mE_X \tr[\lambda^2 X(X^T X +
\lambda I)^{-1} X^T X (X^T X + \lambda I)^{-1} X^T]},
\end{align*}
where the last line holds provided that the numerator and denominator
both converge to finite nonzero limits, as will be confirmed by
our arguments below. We study the numerator first.  Noting that
\smash{$\lambda^2 (X^T X + \lambda I)^{-1}
X^T X (X^T X + \lambda I)^{-1}-X^T X$} has eigenvalues
$$
d_i^2 \bigg(\frac{\lambda^2}{(d_i^2 + \lambda)^2}-1\bigg), \;
i=1,\ldots,p,
$$
we have that \smash{$\lambda^2 (X^T X + \lambda I)^{-1}
X^T X (X^T X + \lambda I)^{-1} \to X^T X$} as $\lambda \to \infty$,
in (say) the operator norm, implying
\smash{$\tr[\lambda^2 X_0 (X^T X + \lambda I)^{-1}
X^T X (X^T X + \lambda I)^{-1} X_0^T] \to \tr(X_0 X^T X X_0^T)$} as
$\lambda \to \infty$.  Hence
\begin{align*}
\lim_{\lambda \to \infty} \mE_{X,X_0} \tr\big[
\lambda^2 X_0(X^T X + \lambda I)^{-1} &X^T X
(X^T X + \lambda I)^{-1} X_0^T\big] \\
&= \mE_{X,X_0} \lim_{\lambda \to \infty} \tr\big[
\lambda^2 X_0(X^T X + \lambda I)^{-1} X^T X
(X^T X + \lambda I)^{-1} X_0^T\big] \\
&= \mE_{X,X_0} \tr(X_0 X^T X X_0^T) \\
&= \tr \big[ \mE_X(X^T X)\mE_{X_0}(X_0^T X_0) \big] \\
&= \tr \big[ \mE_X(X^T X) \mE_X(X^T X)\big].
\end{align*}
Here, in the first line, we applied the dominated convergence theorem
as previously, in the third we used the independence of $X,X_0$, and
in the last we used the identical distribution of $X,X_0$. Similar
arguments lead to the conclusion for the denominator
$$
\lim_{\lambda \to \infty} \mE_X \tr\big[
\lambda^2 X(X^T X + \lambda I)^{-1} X^T X
(X^T X + \lambda I)^{-1} X^T\big] =
\tr \big[ \mE_X(X^T XX^T X)\big],
$$
and thus we have shown that
$$
\lim_{\lambda \to \infty}
\frac{\mE_{X,X_0} \tr[ \Cov(\hf_n(X_0) \,|\, X,X_0)]}
{\mE_X \tr[ \Cov(\hf_n(X) \,|\, X)]} =
\frac{\tr [ \mE_X(X^T X) \mE_X(X^T X)]}
{\tr [ \mE_X(X^T XX^T X)]},
$$
as desired. To see that the ratio on the right-hand side is at most 1,
consider
$$
A = \mE(X^TX X^TX) - \mE(X^TX) \mE(X^TX),
$$
which is a symmetric matrix whose trace is
$$
\tr(A) = \sum_{i,j=1}^p \Var \big((X^T X)_{i,j}\big) \geq 0.
$$
Furthermore, the trace is zero if and only if all summands are zero,
which occurs if and only if all components of $x \sim Q$ have no
variance.
\end{proof}

In words, the theorem shows that the excess variance \ev{} of ridge
regression approaches zero as $\lambda \to \infty$, but it does so
from the left (negative side) of zero.  As we can have cases in which
the excess bias is very small or even zero (for example, a null model
like in our simulations below), we see that $\ErrR-\ErrS=\ev$ can be
negative for ridge regression with a large level of regularization;
this is a striking contrast to the behavior of this gap for least
squares, where it is always nonnegative.

We finish by demonstrating this result empirically, using a simple
simulation setup with $p=100$ covariates drawn from
$Q=N(0,I)$, and training and test sets each of size $n=300$.  The
underlying regression function was $f(x)=\mE(y|x)=0$, i.e., there was
no signal, and the errors were also standard normal.  We drew
training and test data from this simulation setup, fit ridge
regression estimators to the training at various levels of $\lambda$,
and calculated the ratio of the sample versions of the Random-X and
Same-X integrated variances.  We repeated this 100 times, and averaged
the results. As shown in Figure \ref{fig:ridge_ev_neg}, for values of
$\lambda$ larger than about 250, the Random-X integrated variance
is smaller than the Same-X integrated variance, and consequently the
same is true of the prediction errors (as there is no
signal, the Same-X and Random-X integrated biases are both zero).
Also shown in the figure is the theoretical limiting ratio of the
integrated variances according to Theorem \ref{thm:ridge_ev_neg},
which in this case can be calculated from the properties of Wishart
distributions to be $n^2p/(n^2p+np^2+np) \approx 0.7481$, and is in
very good agreement with the empirical limiting ratio.

\begin{figure}[htb]
\centering
\includegraphics[width=0.65\textwidth]{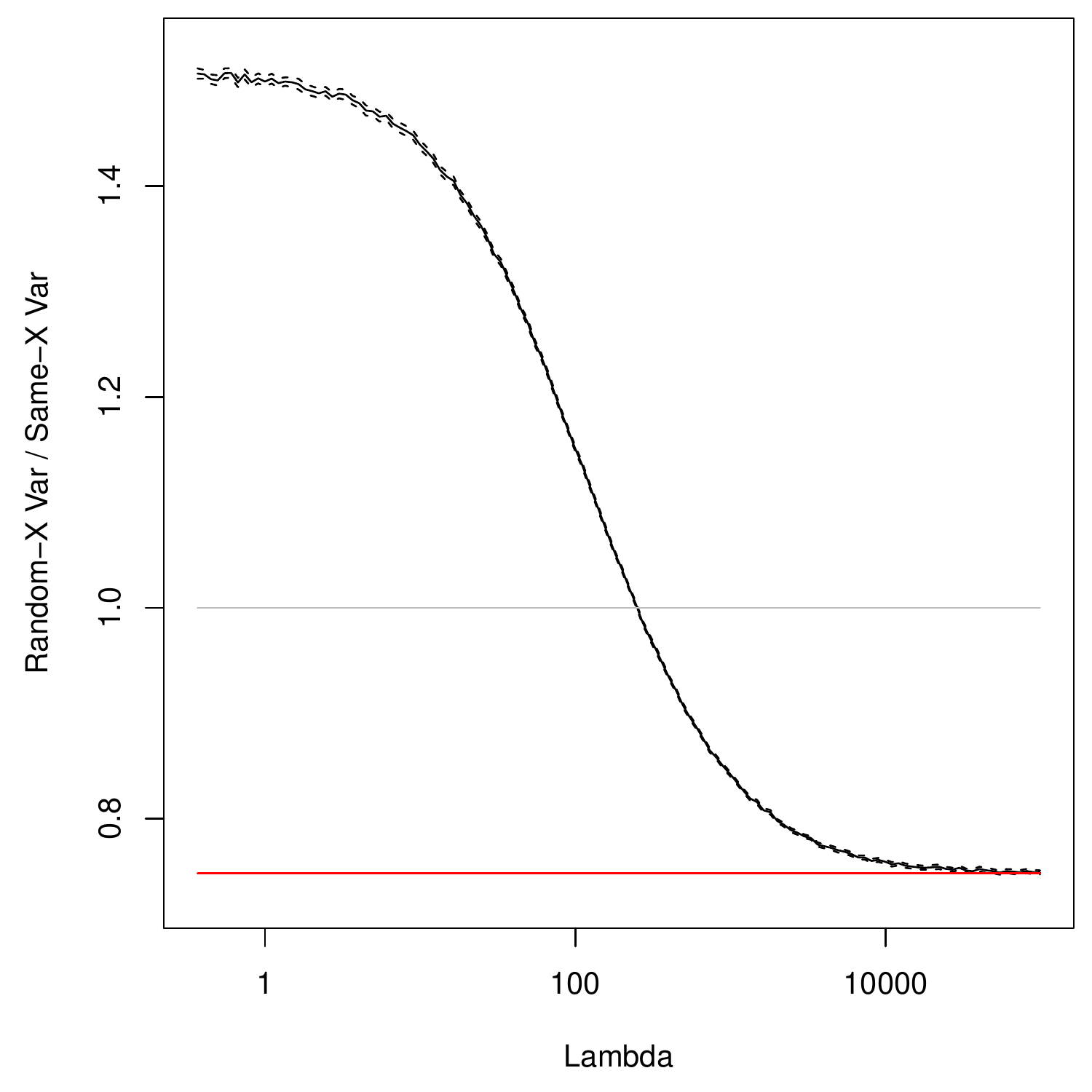}
\caption{\it Ratio of Random-X integrated variance to Same-X
  integrated variance for ridge regression as we vary the tuning
  parameter $\lambda$, in a simple problem setting with
  $p=100$ covariates and training and test sets each of size $n=300$.
  For values of $\lambda$ larger than about
  250, the ratio is smaller than 1, meaning the Random-X prediction
  variance is smaller than the Same-X prediction variance; as the
  integrated bias is zero in both settings, the same ordering also
  applies to the Random-X and Same-X prediction errors.
  The plotted curve is an average of the computed ratios over
  100 repetitions, and the
  dashed lines (hard to see, because they are very close to the
  aforementioned curve) denote 95\% confidence intervals over these
  repetitions.  The red line is the theoretical limiting ratio of
  integrated variances due to Theorem \ref{thm:ridge_ev_neg}, in good
  agreement with the simulation results.}
\label{fig:ridge_ev_neg}
\end{figure}

\section{Nonparametric regression estimators}
\label{sec:nonpar}

We present a brief study of the excess bias and variance of some
common nonparametric regression estimators.  In Section
\ref{sec:discussion}, we give a high-level discussion of the view on
the gap between
Random-X and Same-X prediction errors from the perspective of
empirical process theory, which is a topic that is well-studied by
researchers in nonparametric regression.

\subsection{Reproducing kernel Hilbert spaces}

Consider an estimator \smash{$\hf_n$} defined by the general-form
functional optimization problem
\begin{equation}
\label{eq:rkhs}
\hf_n = \argmin_{g \in \mathcal{G}} \; \sum_{i=1}^n
\big(y_i-g(x_i)\big)^2 + J(g),
\end{equation}
where $\mathcal{G}$ is a function class and $J$ is a roughness penalty
on functions.  Examples estimators of this form include the (cubic)
smoothing spline estimator in $p=1$ dimensions, in which $\mathcal{G}$
is the space of all functions that are twice differentiable and whose
second derivative is in square integrable, and \smash{$J(g)=\int
  g''(t)^2 \,dt$}; and more broadly, reproducing kernel Hilbert space
or RKHS estimators (in an arbitrary dimension $p$), in which
$\mathcal{G}$ is an RKHS and \smash{$J(g)=\|f\|_{\mathcal{G}}$} is the
corresponding RKHS norm.

Provided that \smash{$\hf_n$} defined by \eqref{eq:rkhs} is a linear
smoother, which means \smash{$\hf_n(x)=s(x)^T Y$} for a weight
function $s:\mR^p \to \mR$ (that can and will generally also depend on
$X$), we now show that the excess bias of \smash{$\hf_n$} is always
nonnegative.  We note that this result applies to smoothing
splines and RKHS estimators, since these are linear smoothers; it also
covers ridge regression, and thus generalizes the result in Theorem
\ref{thm:ridge_eb_nonneg}.

\begin{thm}
\label{thm:rkhs_eb_nonneg}
For \smash{$\hf_n$} a linear smoother defined by a problem of the form
\eqref{eq:rkhs}, we have $\eb \geq 0$.
\end{thm}

\begin{proof}
Let us introduce a test covariate matrix
$X_0 \in \mR^{n\times p}$ and associated response vector
$Y_0\in\mR^n$, and write the excess bias in \eqref{eq:eb}
as
$$
\eb = \mE_{X,X_0} \frac{1}{n} \big\|
\mE\big(\hf_n(X_0) \,|\, X,X_0 \big)-f(X_0)\big\|_2^2 -
\mE_X \frac{1}{n} \big\|
\mE\big(\hf_n(X) \,|\, X \big)-f(X)\big\|_2^2.
$$
Writing \smash{$\hf_n(x)=s(x)^T Y$} for a weight function
$s:\mR^p\to\mR$, let $S(X) \in \mR^{n\times n}$ be a smoother matrix
that has rows $s(x_1),\ldots,s(x_n)$. Thus \smash{$\hf_n(X)=S(X) Y$},
and by linearity, \smash{$\mE(\hf_n(X) \,|\, X )=S(X)f(X)$}. This in
fact means that we can express \smash{$g_n=\mE(\hf_n|X )$}, a
function defined by $g(x)=s(x)^T f(X)$, as the solution of an
optimization problem of the form \eqref{eq:rkhs},
$$
g_n = \argmin_{g \in \mathcal{G}} \; \|f(X)-g(X)\|_2^2+ J(g),
$$
where we have rewritten the loss term in a more convenient notation.
Analogously, if we denote by \smash{$\hf_{0n}$} the estimator of the
form \eqref{eq:rkhs}, but fit to the test data $X_0,Y_0$ instead of
the training data $X,Y$, and \smash{$g_{0n}=\mE(\hf_{0n}|X_0)$},
then
$$
g_{0n} = \argmin_{g \in \mathcal{G}} \; \|f(X_0)-g(X_0)\|_2^2
+ J(g).
$$
By virtue of optimality of \smash{$g_{0n}$} for the problem in
the last display, we have
$$
\|g_n(X_0) - f(X_0)\|_2^2 + J(g_n) \geq
\|g_{0n}(X_0) - f(X_0)\|_2^2 + J(g_{0n})
$$
and taking an expectation over $X,X_0$ gives
\begin{align*}
\mE_{X,X_0} \Big[\|g_n(X_0) - f(X_0)\|_2^2 + J(g_n)\Big]
&\geq \mE_{X_0} \Big[\|g_{0n}(X_0) - f(X_0)\|_2^2 +
  J(g_{0n})\Big] \\
&= \mE_X \Big[\|g_n(X) - f(X)\|_2^2 + J(g_n)\Big],
\end{align*}
where in the equality step we used the fact that $X,X_0$ are identical
in distribution. Cancelling out the common term of \smash{$\mE_X
  J(g_n)$} from the the first and third expressions proves the
  result, because \smash{$\mE(\hf_n(X_0) \,|\, X,X_0) =
    g_n(X_0)$} and \smash{$\mE(\hf_n(X) \,|\,
    X)=g_n(X)$}.
\end{proof}

\subsection{$k$-nearest-neighbors regression}

Consider \smash{$\hf_n$} the $k$-nearest-neighbors or kNN regression
estimator, defined by
$$
\hf_n(x) = \frac{1}{k} \sum_{i \in N_k(x)} y_i,
$$
where $N_k(x)$ returns the indices of the $k$ nearest points among
$x_1,\ldots,x_n$ to $x$.  It is immediate that the excess variance of
kNN regression is zero.

\begin{prop}
\label{prop:knn_ev_zero}
For \smash{$\hf_n$} the kNN regression estimator, we have $\ev=0$.
\end{prop}

\begin{proof}
Simply compute
$$
\Var\big( \hf_n(x_1) \,| X \big) = \frac{\sigma^2}{k},
$$
by independence of $y_1,\ldots,y_n$, and hence the points in the
nearest neighbor set $N_k(x_1)$, conditional on $X$.  Similarly,
$$
\Var\big( \hf_n(x_0) \,| X,x_0 \big) = \frac{\sigma^2}{k}.
$$
\end{proof}

On the other hand, the excess bias is not easily computable, and is
not covered by Theorem \ref{thm:rkhs_eb_nonneg}, since kNN
cannot be written as an estimator of the form \eqref{eq:rkhs} (though
it is a linear smoother). The next result sheds some light on the
nature of the excess bias.

\begin{prop}
\label{prop:knn_eb_exact}
For \smash{$\hf_n$} the kNN regression estimator, we have
$$
\eb = B_{n,k} - \bigg(1-\frac{1}{k^2}\bigg) B_{n-1,k-1},
$$
where $B_{n,k}$ denotes the integrated Random-X prediction bias of
the kNN estimator fit to a training set of size $n$, and with tuning
parameter (number of neighbors) $k$.
\end{prop}

\begin{proof}
Observe
$$
\Big(\mE\big(\hf_n(x_0)\,|\,X,x_0\big)-f(x_0)\Big)^2=
\Bigg( \frac{1}{k} \sum_{i \in N_k(x_0)}
\big(f(x_i)-f(x_0)\big)\Bigg)^2,
$$
and by definition,
\smash{$\mE_{X,x_0}[\mE(\hf_n(x_0)\,|\,X,x_0)-f(x_0)]^2=B_{n,k}$}.
Meanwhile
$$
\Big(\mE\big(\hf_n(x_1)\,|\,X\big)-f(x_1)\Big)^2=
\Bigg( \frac{1}{k} \sum_{i \in N_k(x_1)}
\big(f(x_i)-f(x_1)\big)\Bigg)^2
= \Bigg( \frac{1}{k} \sum_{i \in N_k^{-1}(x_1)}
\big(f(x_i)-f(x_1)\big)\Bigg)^2,
$$
where \smash{$N_{k-1}^{-1}(x_1)$} gives the indices of the $k-1$
nearest points among $x_2,\ldots,x_n$ to $x_1$ (which equals
$N_k(x_1)$ as $x_1$ is trivially one of
its own $k$ nearest neighbors).  Now notice that $x_1$ plays the role
of the test point $x_0$ in the last display, and therefore,
\smash{$\mE_X[\mE(\hf_n(x_1)\,|\,X)-f(x_1)]^2=((k-1)/k)^2B_{n-1,k-1}$}.
This proves the result.
\end{proof}

The above proposition suggests that, for moderate values of $k$, the
excess bias in kNN regression is likely positive.
We are comparing the integrated Random-X bias of a kNN model with $n$
training points and $k$ neighbors to that of a model $n-1$ points and
$k-1$ neighbors; for large $n$ and moderate $k$, it seems that the
former should be larger than the latter, and in addition, the
factor of $(1-1/k^2)$ multiplying the latter term makes it even more
likely that the difference \smash{$B_{n,k} -
  (1-1/k^2)B_{n-1,k-1}$} is positive.  Rephrased, using
the zero excess variance result of Proposition
\ref{prop:knn_ev_zero}: the gap in Random-X and Same-X
prediction errors, \smash{$\ErrR-\ErrS=B_{n,k} -
  (1-1/k^2)B_{n-1,k-1}$}, is likely positive for large $n$ and
moderate $k$.  Of course, this is not a formal proof; aside from the
choice of $k$, the shape of the underlying mean function
$f(x)=\mE(y|x)$ obviously plays an important role here too.  As a
concrete problem setting, we might try analyzing the Random-X bias
$B_{n,k}$ for $f$ Lipschitz and a scaling for $k$ such
that $k \to \infty$ but $k/n \to 0$ as $n \to \infty$, e.g., \smash{$k
  \asymp \sqrt{n}$}, which ensures consistency of kNN.  Typical
analyses provide upper bounds on the kNN bias in this problem
setting (e.g., see \citealt{gyorfi2002distribution}), but a more
refined analysis would be needed to compare $B_{n,k}$ to
$B_{n-1,k-1}$.

\section{Discussion}
\label{sec:discussion}

We have proposed and studied a division of Random-X prediction error
into components: the irreducible error $\sigma^2$, the traditional
(Fixed-X or Same-X) integrated bias $B$ and integrated variance $V$
components, and our newly defined excess bias \eb{} and excess
variance \ev{} components, such that $B+\eb$ gives the Random-X
integrated bias and
$V+\ev$ the Random-X integrated variance.  For least squares
regression, we were able to quantify
\ev{} exactly when the covariates are normal and asymptotically when
they are drawn from a linear transformation of a product distribution,
leading to our definition of \rcp.  To account for unknown error
variance $\sigma^2$, we defined \smash{\hrcp} based on the usual
plug-in estimate, which turns out to be asymptotically identical to
GCV, giving this classic method a
novel interpretation. To account for \eb{} (when $\sigma^2$ is known
and the distribution $Q$ of the covariates is well-behaved),
we defined \rcpp{}, by leveraging a Random-X bias estimate implicit to
OCV. We also briefly considered methods beyond least squares, proving
that \eb{} is nonnegative in all settings considered, while \ev{}
can become negative in the presence of heavy regularization.

We reflect on some issues surrounding our findings and possible
directions for future work.

\paragraph{Ability of \smash{\hrcp{}} to account for bias.}
An intriguing phenomenon that we observe is the ability of
\smash{\hrcp}/Sp and its close (asymptotic) relative GCV to deal to
some extent with \eb{} in estimating Random-X prediction error,
through the inflation it performs on the squared training residuals.
For GCV in particular, where recall
\smash{$\mathrm{GCV} = \mathrm{RSS}/(n(1-\gamma)^2)$}, we see that
this inflation a simple form: if the linear model is biased, then
the squared bias component in each residual is
inflated by $1/(1-\gamma)^2$. Comparing this to the inflation that OCV
performs, which is \smash{$1/(1-h_{ii})^2$}, on the $i$th residual,
for $i=1,\ldots,n$, we can interpret GCV as inflating the bias for
each residual by some ``averaged'' version of the elementwise factors
used by OCV.  As OCV provides an almost-unbiased estimate of \eb{} for
Random-X prediction, GCV can get close when the diagonal elements
$h_{ii}$, $i=1,\ldots,n$ do not vary too wildly. When they do vary
greatly, GCV can fail to account for \eb,  as in the circled region in
Figure \ref{fig:relative_low}.


\paragraph{Alternative bias-variance decompositions.}
The integrated terms we defined are expectations of conditional
bias and variance terms, where we conditioned on both training and
testing covariates $X,x_0$. One could also consider other
conditioning schemes, leading to different decompositions. An
interesting option would be to condition on the prediction point $x_0$
only and calculate the bias and variance unconditional on the training
points $X$ before integrating, as in
\smash{$\mE_{x_0}(\mE(\hf_n(x_0)\,|\,x_0)-f(x_0))^2$} and
\smash{$\mE_{x_0}(\Var(\hf_n(x_0)\,|\,x_0))$} for these alternative
notions of Random-X bias and variance, respectively.  It is easy to
see that this would cause the bias (and thus excess bias) to
decrease and variance (and thus excess variance) to increase.
However, it is not clear to us that computing or bounding such new
definitions of (excess) bias and (excess) variance would be possible
even for least squares regression. Investigating the tractability of
this approach and any insights it might offer is an interesting topic
for future study.

\paragraph{Alternative definitions of prediction error.}
The overall goal in our work was to estimate the prediction error,
defined as
\smash{$\ErrR=\mE_{X,Y,x_0,y_0}(y_0-\hf_n(x_0))^2$}, the squared
error integrated over all of the random variables available in
training and testing. Alternative definitions have been suggested by
some authors. \citet{breiman1992submodel}
generalized the Fixed-X setting in a manner that led them to define
\smash{$\mE_{Y,x_0,y_0}[(y_0-\hf_n(x_0))^2 \,|\, X]$} as the
prediction error quantity of interest, which can be interpreted as the
Random-X prediction error of a Fixed-X model.
\citet{hastie2009elements} emphasized the importance of the
quantity 
\smash{$\mE_{x_0,y_0}[(y_0-\hf_n(x_0))^2 \,|\, X,Y]$}, which is the
out-of-sample error of the specific model we have trained on the given
training data $X,Y$.  Of these two alternate definitions, the second
one is more interesting in our opinion, but investigating it
rigorously requires a different approach than what we have developed
here.

\paragraph{Alternative types of cross-validation.}
Our exposition has concentrated on comparing OCV to generalized
covariance penalty methods. We have not discussed other
cross-validation approaches, in particular, K-fold
cross-validation (KCV) method with $K \ll n$ (e.g., $K=5$ or 10). A
supposedly well-known problem with OCV is that its estimates of
prediction error have very high variance; we indeed
observe this phenomenon in our simulations (and for least squares
estimation, the analytical form of OCV clarifies the source of this
high variance). There are some claims in the
literature that KCV can have lower variance than OCV
(\citealt{hastie2009elements}, and others), and should be
considered as the preferred CV variant for estimation of Random-X
prediction error. Systematic investigations of this issue for least
squares regression such as
\citet{burman1989comparative,arlot2010survey} actually reach the
opposite conclusion---that high variance is further compounded by
reducing $K$. Our own simulations also support this view (results not
shown), therefore we do not consider KCV to be an important benchmark
to consider beyond OCV.

\paragraph{Model selection for prediction.}
Our analysis and simulations have focused on the accuracy of
prediction error estimates provided by various
approaches. We
have not considered their utility for model selection, i.e., for
identifying the best predictive model, which differs from model
evaluation in an important way. A method can do well in the model
selection task even when it is inaccurate
or biased for model evaluation, as long as such inaccuracies are
consistent across different
models and do not affect its ability to select the better predictive
model. Hence the correlation of model evaluations using the same
training data across different models plays a central role in model
selection performance. An investigation of the correlation between
model evaluations that each of the approaches we considered here
creates is of major interest, and is left to future work.

\paragraph{Semi-supervised settings.}
Given the important role that the marginal distribution $Q$ of $x$
plays in evaluating Random-X prediction error (as expressed,
e.g., in Theorems \ref{thm:least_squares_normal} and
\ref{thm:least_squares_asymp}), it is of interest to consider
situations where, in addition to the training data, we have large
quantities of additional observations with $x$ only and no response
$y$. In the machine learning literature this situation is often
considered under then names semi-supervised learning or transductive
learning. Such data could be used, e.g., to directly estimate the
excess variance from expressions like \eqref{eq:integrated_var}.

\paragraph{General view from empirical process theory.}
This paper was focused in large part on estimating or
bounding the excess bias and variance in specific problem settings,
which led to estimates or bounds on the gap in Random-X and
Same-X prediction error, as $\ErrR-\Err=\eb+\ev$.  This gap is indeed
a familiar concept to those well-versed in the theory of nonparametric
regression, and roughly speaking, standard results from empirical
process theory suggest that we should in general expect $\ErrR-\ErrS$
to be small, i.e., much smaller than either of \ErrR{} or \ErrS{} to
begin with.  The connection is as follows. Note that
\begin{align*}
\ErrR - \ErrS &= \mE_{X,Y,x_0} \big( f(x_0) - \hf_n(x_0)\big)^2 -
 \mE_{X,Y} \bigg[\frac{1}{n} \sum_{i=1}^n \big( f(x_i) -
 \hf_n(x_i)\big)^2\bigg] \\
&= \mE_{X,Y} \Big[ \| f-\hf_n\|_{L_2(Q)}^2 -
\| f-\hf_n\|_{L_2(Q_n)}^2 \Big],
\end{align*}
where we are using standard notation from nonparametric regression for
``population'' and ``empirical''  norms,
\smash{$\|\cdot\|_{L_2(Q)}$} and \smash{$\|\cdot\|_{L_2(Q_n)}$},
respectively. For an appropriate function class $\mathcal{G}$,
empirical process theory can be used to control the deviations between
\smash{$\|g\|_{L_2(Q)}$} and
\smash{$\|g\|_{L_2(Q_n)}$}, uniformly over all functions $g \in
\mathcal{G}$. Such uniformity is important, because it gives us
control on the difference in population and empirical norms for the
(random) function \smash{$g=f-\hf_n$} (provided of course this
function lies in $\mathcal{G}$).
This theory applies to finite-dimensional $\mathcal{G}$ (e.g., linear
functions, which would be relevant to the  case when $f$ is assumed to
be linear and \smash{$\hf_n$} is chosen to be linear), and even to
infinite-dimensional classes $\mathcal{G}$, provided we have
some understanding of the entropy or Rademacher complexity of
$\mathcal{G}$ (e.g., this is true of Lipschitz functions, which would
be relevant to the analysis of $k$-nearest-neighbors regression or
kernel estimators).

Under appropriate conditions, we typically find
\smash{$\mE_{X,Y}|\|f-\hf_n\|_{L_2(Q)}^2-\|f-\hf_n\|_{L_2(Q_n)}^2|
=O(C_n)$}, where $C_n$ is the
$L_2(Q)$ convergence rate of \smash{$\hf_n$}
to $f$. 
This is even true in an asymptotic setting in which $p$ grows with
$n$ (so $C_n$ here gets replaced by $C_{n,p}$), but such
high-dimensional results usually require more restrictions on the
distribution $Q$ of covariates.
The takeaway message: in most cases where \smash{$\hf_n$} is
consistent with rate $C_n$, we should expect to see the gap being
 \smash{$\ErrR-\ErrS =
  O(C_n)$}, whereas $\ErrR,\ErrS \geq \sigma^2$, so the difference in
Same-X and Random-X prediction error is quite small (as small as the
Same-X and Random-X risk) compared to these prediction errors
themselves; said differently, we should expect to see
$\eb,\ev$ being of the same order (or smaller than) $B,V$.

It is worth pointing out that several interesting aspects of our study
really lie outside what can be inferred from empirical process theory.
One aspect to mention is the precision of the results: in some
settings we can characterize $\eb,\ev$ individually, but (as
described above), empirical process theory would only provide a handle
on their sum. Moreover, for least squares regression estimators, with
$p/n$ converging to a nonzero constant, we are able to characterize
the exact asymptotic excess variance under some conditions on $Q$
(essentially, requiring $Q$ to be something like a rotation of a
product distribution), in Theorem \ref{thm:least_squares_asymp}; note
that this is a problem setting in which least squares is not
consistent, and could not be treated by standard results empirical
process theory.

Lastly, empirical process theory tells us nothing about the sign of
$$
\| f-\hf_n\|_{L_2(Q)}^2 - \| f-\hf_n\|_{L_2(Q_n)}^2,
$$
or its expectation under $P$ (which equals $\ErrR-\ErrS=\eb+\ev$, as
described above). This 1-bit quantity is of interest to us,
since it tells us if the Same-X (in-sample) prediction error is
optimistic compared to the Random-X (out-of-sample) prediction error.
Theorems \ref{thm:least_squares_nonneg}, \ref{thm:ridge_eb_nonneg},
\ref{thm:ridge_ev_neg}, \ref{thm:rkhs_eb_nonneg} and Propositions
\ref{prop:knn_ev_zero}, \ref{prop:knn_eb_exact} all pertain to this
quantity.


\bibliographystyle{plainnat}

\end{document}